\pdfoutput=1
\documentclass[11pt,a4paper]{article}

\usepackage[margin=1in]{geometry}
\usepackage{amsfonts}
\usepackage{amsmath,amssymb,amsthm,mathtools,booktabs}
\usepackage{authblk}
\usepackage{enumitem}
\usepackage{etoolbox}
\usepackage{microtype}
\usepackage{hyperref}

\setenumerate[1]{label={(\alph{enumi})}}
\AtBeginEnvironment{thebibliography}{\footnotesize\setlength{\itemsep}{0.1em}}

\theoremstyle{plain}
\newtheorem{theorem}{Theorem}
\newtheorem{proposition}[theorem]{Proposition}
\newtheorem{lemma}[theorem]{Lemma}
\newtheorem{corollary}[theorem]{Corollary}
\newtheorem{conjecture}[theorem]{Conjecture}
\theoremstyle{definition}
\newtheorem{remark}[theorem]{Remark}
\newtheorem{observation}[theorem]{Observation}
\newtheorem{definition}{Definition}

\newcommand{\OPT}{\mathrm{OPT}}

\newcommand{\cost}{\mathrm{cost}}
\newcommand{\Off}{\mathrm{Off}}
\newcommand{\FtP}{\textnormal{\textsc{FtP}}}
\newcommand{\Sd}{\mathcal{A}}
\newcommand{\eps}{\varepsilon}
\newcommand{\R}{\mathbb{R}}

\newif\ifanon
\anonfalse

\title{The Price of Near-Perfect Consistency in\\ Online Metric Matching with Predictions}
\ifanon
  \author{\normalsize\itshape Author name withheld for double-blind review}
  \affil{}
  \newcommand{\pdfauthorstring}{}
\else
  \author{Zaahir Ali}
  \affil{University of Warwick}
  \newcommand{\pdfauthorstring}{Zaahir Ali}
\fi
\date{}
\hypersetup{
  colorlinks=true,
  linkcolor=blue,
  citecolor=blue,
  urlcolor=blue,
  pdftitle={The Price of Near-Perfect Consistency in Online Metric Matching with Predictions},
  pdfauthor={\pdfauthorstring},
  pdfsubject={Online algorithms with predictions},
  pdfkeywords={online metric matching, learning-augmented algorithms, consistency, robustness}
}

\begin{document}

\maketitle

\begin{abstract}
We study online metric matching with per-request action predictions. On the real line, every
deterministic $(1+\eps)$-consistent algorithm has robustness at least
$1+\sum_{j=1}^{k-1}2^{j+1}/\eps^j$, and we give a deterministic algorithm for arbitrary metrics
with the same leading term. Thus, for every fixed $k$,
\[
\lim_{\eps\downarrow0}\eps^{k-1}R_k^{\R}(1+\eps)
=\lim_{\eps\downarrow0}\eps^{k-1}R_k(1+\eps)=2^k.
\]
The comparison is uniform up to an absolute constant for $0<\eps\le1/(k-1)$. We determine the
two-server trade-off in both settings and the real-line three-server value
$1+4/\eps+8/\eps^2$ for $0<\eps\le\sqrt{13}-3$. When the predicted labels are distinct, the
algorithm pays at most $(1+\eps)$ times the cost of the predicted matching.

For randomised algorithms, the fixed-$k$ dependence remains
$\Theta_k(1/\eps^{k-1})$. Uniformly in $k$, robustness is at most
$M_0(\eps)\rho_k^0$, where $\rho_k^0$ is the optimal strict prediction-free randomised ratio on the
real line and $M_0(\eps)=(2e+o(1))e^{2/\eps}$. For every $\eta>0$, a lower bound
$\exp((2-\eta)/\eps)$ holds once $k\ge C_\eta/\eps$.

The randomised upper bound follows from a comparison theorem for two online algorithms whose states
can be coupled at a cost bounded by their cumulative costs. For every fixed $c>1$, the least
comparison factor $M^\star(c,\eps)$ under these assumptions satisfies
\[
\lim_{\eps\downarrow0}\eps\log M^\star(c,\eps)=2.
\]
The guarantee is strictly multiplicative and has no diameter-dependent additive term. Irrevocable
metric matching and metrical task systems satisfy the assumptions, and the exponent $2$ is optimal
under them.
\end{abstract}

\section{Introduction}

In online metric matching, a collection
$S=\{s_1,\dots,s_k\}$ of labelled \emph{server} copies in a metric space $(X,d)$ is known in advance.
Requests $r_1,\dots,r_k\in X$ arrive one by one. Each request must be matched immediately and
irrevocably to a currently free server, at cost $d(r_t,\,\cdot\,)$. The benchmark $\OPT$ is the
offline minimum-cost perfect matching. Deterministic algorithms cannot beat competitive ratio
$2k-1$ in general metrics~\cite{KP93,KMV94}. On the real line, the best known deterministic algorithm
achieves $O(\log k)$~\cite{Rag18}, and every randomised algorithm has ratio
$\Omega(\sqrt{\log k})$~\cite{PS23}.

In the \emph{action-prediction}
model of Antoniadis, Coester, Eli\'a\v{s}, Polak and Simon~\cite{ACEPS23}, each request $r_t$
arrives with a predicted server $p_t\in S$, the action an offline optimum would take. An
algorithm is \emph{$c$-consistent} if its cost is at most $c\cdot\OPT$ whenever the predictions
are perfect (they agree with an offline optimum), and \emph{$r$-robust} if its cost is at most
$r\cdot\OPT$ for arbitrary predictions. Following the prediction blindly ($\FtP$) is
$1$-consistent and has unbounded robustness. The combiner of~\cite{ACEPS23} gives consistency $9$
with robustness $O(\log k)$ on the real line. To our
knowledge, no previous lower bound addresses the intermediate consistency band in this
per-request action-prediction model. Canonne, Chen and Mestre~\cite{CCM25} give a related lower
bound at exact consistency in a broader advice setting. Their endpoint construction does not yield
a nontrivial bound once the consistency factor exceeds $1$.

We address the question at two levels. For metric matching, we determine how robustness grows as a
deterministic or randomised algorithm approaches perfect consistency. For the randomised upper
bound, we isolate the comparison problem that remains after the matching states can be coupled: how
closely can one follow a prediction-based trajectory while retaining a bounded comparison with a
competitive reference? The comparison theorem answers this question for any online problem that
satisfies the same switching condition. Metric matching and metrical task systems both satisfy it.

We determine the near-perfect-consistency exponent for every fixed $k$. The bounds are exact for
two servers in both settings and for three servers on the real line when
$0<\eps\le\sqrt{13}-3$. For fixed $k$, every $(1+\eps)$-consistent deterministic algorithm on the
real line needs robustness of order $1/\eps^{k-1}$, with leading coefficient $2^k$. Each additional
server therefore contributes a factor asymptotic to $2/\eps$ to the final term. In particular, an
additive $O(1/\eps)+O(\log k)$ bound is impossible.

For randomised algorithms the fixed-$k$ exponent remains $k-1$ in expectation. When $\eps$ is fixed
and $k$ grows, robustness is at most
$M_0(\eps)\rho_k^0$, where $\rho_k^0$ is the optimal strict prediction-free ratio and
$M_0(\eps)=(2e+o(1))e^{2/\eps}$. A lower bound $e^{(2-o(1))/\eps}$ holds once
$k=\Theta(1/\eps)$. The upper bound follows by coupling the states of two online algorithms and
has no diameter-dependent additive loss. Metrical task systems and irrevocable metric matching
satisfy the required assumptions. Let $R_k(c)$ be the infimum robustness over deterministic
algorithms that are defined, $c$-consistent, and robust on every metric space. Let $R_k^{\R}(c)$ be the
corresponding infimum when the metric is restricted to the real line. Since every general-metric
algorithm applies on the real line,
\[
R_k^{\R}(c)\le R_k(c).
\]

\subsection{Results}

\paragraph{Two servers.}
\begin{theorem}[two-server trade-off; Theorem~\ref{thm:k2}]
For every $\eps\in(0,2]$,
$R_2^{\R}(1+\eps)=R_2(1+\eps)=(4+\eps)/\eps$. On the real line, the optimum is achieved by a
threshold algorithm that selects the nearer server when the request is within $\eps/(2+\eps)$ of
it, and selects the predicted server otherwise.
\end{theorem}

\paragraph{Lower bound for general $k$.}
The $k=2$ dependence $\Theta(1/\eps)$ is the first term of a longer sum.

\begin{theorem}[lower bound for general $k$; Theorem~\ref{thm:lower-general}]\label{intro:lower-general}
For every $k\ge2$ and every $\eps>0$, every $(1+\eps)$-consistent deterministic algorithm for $k$
servers on the real line has robustness at least
\[
L_k(\eps)\;=\;1+\sum_{j=1}^{k-1}\frac{2^{\,j+1}}{\eps^{\,j}}
\;=\;1+\frac4\eps+\frac8{\eps^2}+\dots+\frac{2^k}{\eps^{k-1}}.
\]
\end{theorem}

Each additional server contributes a factor $2/\eps$ to the final term. The bound is tight at
$k=2$.

\begin{corollary}[no additive trade-off]\label{intro:additive}
There is no deterministic algorithm that is $(1+\eps)$-consistent with robustness
$O(1/\eps)+O(\log k)$, and none with robustness $f(\eps)+g(k)$ for any
$f(\eps)=o(1/\eps^2)$. Consistency $9$ with robustness $O(\log k)$ is
achievable~\cite{ACEPS23}.
\end{corollary}

For every fixed $c<3$, Theorem~\ref{intro:lower-general} gives exponential growth in $k$. At
$c=3$ it gives $R_k^{\R}(3)\ge2k-1$. The endpoint behaviour beyond this lower bound remains open.

\paragraph{General metrics.}
A single deterministic algorithm attains the lower-bound exponent for every fixed $k$ on arbitrary
metrics.

\begin{theorem}[general-metric upper bound; Theorems~\ref{thm:consistency},
\ref{thm:fixedkleading}, and Lemma~\ref{lem:closedform}]
For every $k$ and every metric space, the deterministic algorithm $\Sd_\eps$ is
$(1+\eps)$-consistent. For $\eps\in(0,1]$ it satisfies
\[
\cost(\Sd_\eps)
\le\left(1+2\left(\frac2\eps+3\right)^{k-1}\right)\OPT
\]
on every input. If the predicted server labels are pairwise distinct, then
\[
\cost(\Sd_\eps)
\le(1+\eps)\sum_{t=1}^k d(r_t,p_t).
\]
Consequently, for every fixed $k$,
\[
\lim_{\eps\downarrow0}\eps^{k-1}R_k^{\R}(1+\eps)
=\lim_{\eps\downarrow0}\eps^{k-1}R_k(1+\eps)=2^k.
\]
\end{theorem}

If the predicted labels are distinct, they define a perfect matching of cost
$\sum_t d(r_t,p_t)=\OPT+\eta$, where $\eta\ge0$, and the second comparison becomes
$\cost(\Sd_\eps)\le(1+\eps)(\OPT+\eta)$, giving linear degradation in this error measure.

\begin{corollary}[uniform range]\label{intro:uniform}
For every $k\ge2$ and every $0<\eps\le1/(k-1)$,
\[
\frac{2^k}{\eps^{k-1}}
\le R_k^{\R}(1+\eps)
\le R_k(1+\eps)
\le1+e^{3/2}\frac{2^k}{\eps^{k-1}}.
\]
Thus both optimal robustness values are of order $2^k/\eps^{k-1}$ on this range, with constants
independent of $k$.
\end{corollary}

\paragraph{Three servers.}
A separate analysis of the first request gives the leading coefficient when $k=3$.

\begin{theorem}[$k=3$ bounds; Theorems~\ref{thm:k3ub} and~\ref{thm:k3exact}]\label{intro:k3}
For every $\eps>0$ the expression $1+4/\eps+8/\eps^2$ is a lower bound on $R_3^{\R}(1+\eps)$, and
it is the exact value for every $\eps$ with $0<\eps\le\sqrt{13}-3$:
\[
R_3^{\R}(1+\eps)=1+\frac4\eps+\frac8{\eps^2}
\qquad\bigl(0<\eps\le\sqrt{13}-3=0.6055\ldots\bigr).
\]
For $\sqrt{13}-3<\eps\le1$ the value lies between $1+4/\eps+8/\eps^2$ and
$1+12/\eps+8/\eps^2$. In particular, $\lim_{\eps\downarrow0}\eps^2R_3^{\R}(1+\eps)=8$. We
conjecture that the first
expression is the value throughout $0<\eps\le1$ (Conjecture~\ref{conj:k3}).
\end{theorem}

\paragraph{Randomised algorithms.}
Against an oblivious adversary, with consistency and robustness measured in expectation, write
$R_k^{\mathrm{rand},\R}(c)$ for the corresponding optimal robustness.

\begin{theorem}[randomised bounds; Theorems~\ref{thm:randk2}, \ref{thm:randk},
and Corollary~\ref{cor:randfixedband}]
For $\eps\in(0,1]$, $R_2^{\mathrm{rand},\R}(1+\eps)=1+1/\eps$ (and the value is $2$ for
$\eps\in[1,2]$). More generally, for every fixed $k$,
\[
R_k^{\mathrm{rand},\R}(1+\eps)=\Theta_k(1/\eps^{k-1})\qquad(\eps\downarrow0).
\]
Uniformly in $k$,
\[
R_k^{\mathrm{rand},\R}(1+\eps)
 \le \min\!\left\{\frac{6^{k-1}}{\eps^{k-1}},
 e^{(2+o(1))/\eps}\,O(\log k)\right\}.
\]
Conversely, for every $\eta>0$ and all sufficiently small $\eps$, once
$k\ge C_\eta/\eps$ one has
$R_k^{\mathrm{rand},\R}(1+\eps)\ge e^{(2-\eta)/\eps}$.
The power $k-1$ remains necessary for fixed $k$. For every fixed positive $\eps$, the uniform upper
bound is subexponential in $k$.
\end{theorem}

\subsection{Techniques}

\paragraph{The deterministic lower bound.}
All deterministic lower bounds use the following adversarial construction. Suppose the algorithm
has followed predictions of total length $C$, and a request $y$ arrives with predicted server
$p$. If the algorithm uses another free server $w$, the adversary completes the instance with an
exact hit at $w$ and zero-cost padding. The predicted pairs form an optimal matching of
cost $C+d(y,p)$, while the alternative assignment has excess
$\Delta=d(y,w)+d(w,p)-d(y,p)$. This assignment violates consistency when
\begin{equation}\label{eq:deviationthreshold}
\Delta \;>\; \eps\,\bigl(C+d(y,p)\bigr).
\end{equation}
Here $C$ is the predicted cost accumulated before the current request, so the opening request sets
the first constraint in~\eqref{eq:deviationthreshold}. At each later stage, the request lies at the
server used in the preceding stage and its predicted server lies farther to the right, allowing
the distance to increase by a factor asymptotic to $2/\eps$. After $k-1$ such assignments the
final request is at distance $\Theta((2/\eps)^{k-2})$ from the remaining server while $\OPT=x$.
An exact-hit completion at the chosen server, together with Lemma~\ref{lem:residual}, handles
assignments to any other free server.

\paragraph{The deterministic upper bound.}
The algorithm $\Sd_\eps$ maintains
$B_t=\eps\sum_{j\le t}d(r_j,p_j)-\sigma_{t-1}$, where $\sigma_{t-1}$ is the total assignment excess charged
before request $t$. On a perfect instance, $\sum_j d(r_j,p_j)=\OPT$, so this quantity bounds the
total excess cost by $\eps\OPT$. The proof needs only the triangle inequality, applied through the
potential function $\Phi_t$ defined by the minimum matching cost between the algorithm's free
servers and the residual set of predicted servers. For robustness, rejection of the server used
by an offline optimum bounds both the current predicted distance and the accumulated predicted
cost. This rejection bound is applied to the first $k-1$ decisions, after which only one server is
free, yielding a recurrence with base $2+4/\eps$ and an upper bound
$6^{k-1}/\eps^{k-1}$. For $k=3$, a separate first-step estimate gives
$1+12/\eps+8/\eps^2$.

\paragraph{The randomised comparison theorem.}
The randomised result isolates the property of matching used to combine algorithms. Switching
between two residual matchings costs at most the minimum matching distance between their
free-server sets. A triangle-inequality argument bounds this distance by the two reference costs
accumulated so far. We encode these facts in the switching assumptions and analyse a continuous
interpolation that is implemented using at most two random switches per request. Two potentials
control the comparison with the prediction-following and competitive references. The resulting
asymmetric $(1+\eps,M)$ guarantee is strictly multiplicative, independent of the metric diameter,
and has $M=e^{2/\eps+O(\log(1/\eps))}$. The same analysis applies to metrical task systems.

\subsection{Related work}

Classical online metric matching has tight deterministic ratio $2k-1$ in general
metrics~\cite{KP93,KMV94}. Randomised bounds are $O(\log^2 k)$~\cite{BBGN14} and
$\Omega(\log k)$~\cite{MNP06}. On the real line, the deterministic upper bound is
$O(\log k)$~\cite{Rag18}, and Peserico and Scquizzato proved a randomised lower bound
$\Omega(\sqrt{\log k})$~\cite{PS23}. Closing the gap to $O(\log k)$ is a well-known open problem.

For per-request action predictions, Antoniadis, Coester, Eli\'a\v{s}, Polak and
Simon~\cite{ACEPS23} show that $\FtP$ has cost at most $\Off+2\eta$ against any offline
algorithm $\Off$, and give a deterministic $9$-consistent $O(\log k)$-robust combiner on the
line. Shin and Vajanopath~\cite{SV26} study a parsimonious variant with few queried predictions.
Their lower bounds concern the number of queries. Yang and Yu~\cite{YY26} receive the predicted
request multiset upfront, and
Azar, Panigrahi and Touitou~\cite{APT22} use an input-prediction model. Canonne, Chen and
Mestre~\cite{CCM25} prove an exponential lower bound at exact consistency in a broader advice
setting, including randomised algorithms. Their endpoint construction does not extend to a nontrivial lower bound for a factor
$1+\eps$.

On the maximisation side, trade-offs between consistency and robustness are known for two-stage bipartite
matching~\cite{JM22} and for fractional or integral matching under adversarial and random
arrivals~\cite{CGLB24,BEM25}. Those results use a different objective and arrival model.

Algorithm combination for metric online problems has also been studied under different goals.
The randomised MTS combiner of Antoniadis et al.~\cite{ACEPS23} incurs an additive
$O(D/\eps)$ term for diameter $D$, while their later work on mixing predictions competes with a
dynamic sequence of predictors~\cite{ACEPSMix23}, a different objective from the asymmetric
terminal comparison studied here. Dallot et al.~\cite{DEGPS26} give a general compiler in a probabilistic
corruption model. That model randomises the reliability of each piece of guidance and is distinct
from worst-case consistency and robustness for a fixed prediction sequence.

\paragraph{Comparison with prior work.}
The deterministic algorithm $\Sd_\eps$ is specific to metric matching and determines the power of
$1/\eps$ for every fixed $k$. The randomised upper bound compares two online trajectories under
the switching assumptions in Definition~\ref{def:switching}. For
MTS, the earlier randomised combiner of Antoniadis et al.~\cite{ACEPS23} reaches near-$1$
consistency with polynomial dependence on $1/\eps$, but pays an additive term proportional to the
state-space diameter. Our guarantee is strictly multiplicative and diameter-free, at the price
of an exponential $e^{2/\eps+O(\log(1/\eps))}$ robustness factor. The guarantees are incomparable.
DART~\cite{CSW23} already supplies the strict diameter-free MTS conclusion.

\subsection{Organisation}

Section~\ref{sec:prelim} fixes the model, and Section~\ref{sec:k2} proves the exact two-server
trade-off. Section~\ref{sec:lower} proves the lower bound for general $k$ and a complementary bound for
$c\in(1,3)$, after which Sections~\ref{sec:algorithm} and~\ref{sec:k3} analyse $\Sd_\eps$ and
give the exact real-line value for three servers when $0<\eps\le\sqrt{13}-3$. Section~\ref{sec:randomised}
contains the randomised bounds, the abstract comparison theorem, and its applications to matching
and metrical task systems. Section~\ref{sec:discussion} collects the remaining open problems.

\section{Preliminaries}\label{sec:prelim}

\paragraph{Model.}
A metric space $(X,d)$ contains $k$ labelled servers at fixed points
$s_1,\dots,s_k\in X$. Labels remain distinct when locations coincide. Requests
$r_1,\dots,r_k\in X$ arrive online, and each $r_t$ is matched irrevocably to a currently free
server at cost $d(r_t,\cdot)$. The total cost is compared with the offline minimum-cost perfect
matching $\OPT$. A prediction $p_t\in S$ arrives with each request and is visible before the
algorithm acts. A prediction sequence is \emph{perfect} if some offline optimal matching $M^*$ has
$p_t=M^*(r_t)$ for every $t$. The predictions in a perfect sequence are pairwise distinct, and
$\OPT=\sum_t d(r_t,p_t)$. A deterministic algorithm is
\emph{$c$-consistent} if $\cost\le c\cdot\OPT$ on every instance with a perfect prediction
sequence and \emph{$r$-robust} if $\cost\le r\cdot\OPT$ on every instance. (On instances with
$\OPT=0$, both guarantees require cost $0$. All our algorithms meet this condition, and every
lower-bound witness has $\OPT>0$.) Let $R_k(c)$ be the infimum robustness over deterministic algorithms
defined for every metric space that satisfy both guarantees in every metric space. Let
$R_k^{\R}(c)$ be the corresponding infimum for algorithms restricted to the real line. Thus
$R_k^{\R}(c)\le R_k(c)$. An algorithm contributing to $R_k^{\R}(c)$ must handle every placement
of the $k$ servers on $\R$, and our lower-bound trees choose the placement.

For randomised algorithms the input is fixed independently of the internal random bits
(oblivious adversary), and both guarantees concern expected cost on every fixed input. We write
$R_k^{\mathrm{rand},\R}(c)$ for the infimum robustness of randomised algorithms on the real line.

\begin{observation}[$\FtP$]\label{obs:ftp}
The algorithm that matches $r_t$ to $p_t$ whenever $p_t$ is free is $1$-consistent: on a perfect
instance the $p_t$ are distinct, so it plays $M^*$ throughout and pays $\OPT$. Its robustness is
unbounded already for $k=2$ on the real line.
\end{observation}

\begin{lemma}[uncrossing on the real line]\label{lem:uncross}
If requests $x_1\le\dots\le x_m$ and servers $y_1\le\dots\le y_m$ lie on $\R$, the
order-preserving matching $x_i\mapsto y_i$ has minimum cost.
\end{lemma}

\begin{proof}
For $x\le x'$ and $y\le y'$ the function $t\mapsto|t-y|-|t-y'|$ is nondecreasing, so
$|x-y|+|x'-y'|\le|x-y'|+|x'-y|$. Given any optimal matching, repeatedly swapping the partners of
an inverted pair does not increase the cost and strictly decreases the number of inversions. After
finitely many swaps the matching is order-preserving.
\end{proof}

Throughout the lower bounds, every offline optimum is computed by exhibiting the
order-preserving matching. Whenever we claim that a prediction sequence is perfect, the
corresponding matching $M^*$ is the order-preserving matching for that instance, so its
optimality follows from Lemma~\ref{lem:uncross}.

\begin{proposition}[prediction-free lower bounds]\label{prop:floor}
Every $r$-robust algorithm yields an $r$-competitive prediction-free algorithm (feed it the
constant prediction $p_t\equiv s_1$). Hence on the real line $r=\Omega(\sqrt{\log k})$~\cite{PS23},
and in general metrics $r\ge2k-1$ for deterministic algorithms~\cite{KP93}, regardless of
consistency.
\end{proposition}

\section{Two servers}\label{sec:k2}

For two servers on the real line, normalise $s_0=0$ and $s_1=1$. Translation and scaling preserve ratios.
After $r_1$ is matched the second move is forced, so a deterministic algorithm is a map
$(r_1,p_1)\mapsto\{s_0,s_1\}$. By reflection assume $p_1=s_1$ and write $x=r_1$. Two elementary
suprema, proved by the piecewise-linear case check: for $x\in[0,1]$,
\begin{equation}\label{eq:phipsi}
\sup_{y\in\R}\frac{x+|y-1|}{(1-x)+|y|}=\frac{1+x}{1-x}\ (\text{at }y=0),
\qquad
\sup_{y\in\R}\frac{(1-x)+|y|}{x+|y-1|}=\frac{2-x}{x}\ (\text{at }y=1).
\end{equation}

\begin{theorem}[two-server trade-off]\label{thm:k2}
For every $\eps\in(0,2]$,
\[
R_2^{\R}(1+\eps)=R_2(1+\eps)=\frac{4+\eps}\eps.
\]
On the real line, the value is achieved by the algorithm $P_\eps$: on $(x,s_1)$, match $x$ to
$s_0$ if and only if $x\le x_\eps:=\eps/(2+\eps)$, and match it to $s_1$ otherwise.
\end{theorem}

\begin{proof}
\emph{Lower bound.} Let $A$ be $(1+\eps)$-consistent and $x\in(x_\eps,1)$, $p_1=s_1$. If $A$
matches $x$ to $s_0$, the adversary completes with $(r_2,p_2)=(0,s_0)$: the predicted pairs form
the matching $\{x\to s_1,\,0\to s_0\}$ of cost $1-x$, optimal by Lemma~\ref{lem:uncross}, and $A$
pays $1+x$; since $(1+x)/(1-x)>1+\eps$ for $x>x_\eps$, this contradicts consistency. Thus $A$
matches $x$ to $s_1$
at every such $x$; the adversary then plays $r_2=1$, giving cost $2-x$ against $\OPT=x$
(order-preserving), so its robustness is at least $\sup_{x>x_\eps}(2-x)/x=(4+\eps)/\eps$.

\emph{Upper bound.} $P_\eps$ is $(1+\eps)$-consistent: on a perfect instance with $p_1=s_1$,
$\OPT$ is the swap matching's cost $(1-x)+|y|$, obtained by matching $x$ to $s_1$. Matching $x$
to $s_0$ occurs only when $x\le x_\eps$; for $x<0$ this choice weakly dominates for every completion, and otherwise it has ratio at most
$(1+x)/(1-x)\le1+\eps$ by~(\ref{eq:phipsi}). $P_\eps$ is $((4+\eps)/\eps)$-robust: when it
selects $s_0$, the ratio is at most $(1+x)/(1-x)\le1+\eps\le(4+\eps)/\eps$; otherwise it is at
most $(2-x)/x<(4+\eps)/\eps$ for $x>x_\eps$ (and $\le1$ for $x\ge1$), again
by~(\ref{eq:phipsi}).

This proves the equality for $R_2^{\R}$; the definitions give $R_2^{\R}\le R_2$, while
Theorem~\ref{thm:robustness} specialises at $k=2$ to the general-metric upper bound
$1+4/\eps=(4+\eps)/\eps$, proving the same equality for $R_2$.
\end{proof}

\section{Lower bound for a general number of servers}\label{sec:lower}

\begin{lemma}[remaining matching cost]\label{lem:residual}
Fix any continuation of an execution, and at any point let $Q$ be the multiset of its requests not
yet served and $T$ the set of free servers ($|Q|=|T|$). Every online algorithm's remaining cost is
at least the minimum cost of an offline perfect matching of $Q$ to $T$.
\end{lemma}

\begin{proof}
The algorithm's remaining moves form some perfect matching of $Q$ to $T$.
\end{proof}

\begin{theorem}[lower bound for general $k$]\label{thm:lower-general}
For every $k\ge2$ and every $\eps>0$, every $(1+\eps)$-consistent deterministic algorithm for $k$
servers on the real line has robustness at least
$L_k(\eps)=1+\sum_{j=1}^{k-1}2^{\,j+1}/\eps^{\,j}$.
\end{theorem}

\begin{proof}
Let $A$ be a $(1+\eps)$-consistent deterministic algorithm. Put
\[
n=k-1,
\qquad q=\frac2\eps,
\qquad x_\eps=\frac1{1+q}=\frac\eps{2+\eps}.
\]
Fix $x\in(x_\eps,1)$. Choose server locations
\[
0=a_0<a_1=1<a_2<\cdots<a_n
\]
such that
\begin{equation}\label{eq:lower-scales}
a_j<x+qa_{j-1}
\qquad(2\le j\le n).
\end{equation}
These locations exist for every $q>0$: when $q\ge1$, the upper endpoint exceeds $a_{j-1}$, while
for $q<1$ the inequality $x>1-q$ places the fixed point $x/(1-q)$ of $u\mapsto x+qu$ above $1$.
Choosing each $a_j$ between $a_{j-1}$ and $x+qa_{j-1}$ then keeps the sequence below that fixed
point.

The adversary first issues
\begin{equation}\label{eq:lower-prefix}
(r_1,p_1)=(x,a_1),
\qquad
(r_j,p_j)=(a_{j-1},a_j)
\quad(2\le j\le n).
\end{equation}
We show that $A$ cannot assign a prefix request to $a_0$. Suppose $a_0$ is free before request
$j$ and that $A$ assigns this request to $a_0$. Stop the prefix and complete the input with
\begin{equation}\label{eq:lower-completion}
(0,a_0)
\quad\text{and}\quad
(a_i,a_i)
\quad(j<i\le n).
\end{equation}

For $j=1$, the predicted matching is order-preserving with cost $1-x$, whereas the minimum cost
among full matchings containing $x\mapsto a_0$ is $1+x$; since $x>x_\eps$ implies
$2x>\eps(1-x)$, this action contradicts consistency.

Now let $j\ge2$. The predicted matching is order-preserving with cost $a_j-x$. After constraining
the request $a_{j-1}$ to use $a_0$ and deleting this pair, Lemma~\ref{lem:uncross} shows that the
minimum matching of the remaining requests to the remaining servers is order-preserving;
restoring the constrained pair, the minimum cost of a full matching containing
$a_{j-1}\mapsto a_0$ is
\[
a_j-x+2a_{j-1}.
\]
Equation~\eqref{eq:lower-scales} gives
$2a_{j-1}>\eps(a_j-x)$, which again contradicts consistency. This argument does not depend on the
earlier assignments among $a_1,\ldots,a_n$.

After all $n$ prefix requests, the algorithm has used every server in
$\{a_1,\ldots,a_n\}$ and its prefix cost is at least the order-preserving value $a_n-x$. On the
final pair $(a_n,a_0)$, only $a_0$ is free, so the algorithm pays a further $a_n$, whereas the
order-preserving offline matching assigns $x$ to $a_0$ and each $a_i$ to itself, with $\OPT=x$.
The robustness is therefore at least
\begin{equation}\label{eq:lower-terminal-ratio}
\frac{2a_n-x}{x}.
\end{equation}

For fixed $x$, let the locations approach the recurrence
\[
b_1=1,
\qquad
b_j=x+qb_{j-1}.
\]
It gives $b_n=q^{n-1}+x\sum_{i=0}^{n-2}q^i$. Letting $x$ decrease to $x_\eps$ in
\eqref{eq:lower-terminal-ratio} yields
\[
R_k^{\R}(1+\eps)
\ge -1+2\sum_{i=0}^{n}q^i
=1+2\sum_{i=1}^{n}q^i
=1+\sum_{i=1}^{k-1}\frac{2^{\,i+1}}{\eps^{\,i}}.
\]
Because all inequalities used to choose the locations are strict, the displayed value is their
supremum and remains a lower bound on robustness.
\end{proof}

At $k=2$, Theorem~\ref{thm:lower-general} recovers the lower-bound half of Theorem~\ref{thm:k2}
exactly, and its first truncation is tight; the following three consequences will be used below.

\begin{corollary}\label{cor:values}
$R_3^{\R}(1+\eps)\ge1+4/\eps+8/\eps^2$ and $R_4^{\R}(1+\eps)\ge1+4/\eps+8/\eps^2+16/\eps^3$ for
every $\eps>0$.
\end{corollary}

\begin{corollary}[no additive trade-off]\label{cor:additive}
No deterministic algorithm is $(1+\eps)$-consistent with robustness $f(\eps)+g(k)$ for
$f(\eps)=o(1/\eps^2)$; in particular robustness $O(1/\eps)+O(\log k)$ is impossible for
$(1+\eps)$-consistent algorithms.
\end{corollary}

\begin{proof}
Fix $k\ge3$. By Corollary~\ref{cor:values}, robustness is at least $8/\eps^2$ for all
$\eps\in(0,1]$; as $\eps\downarrow0$ this exceeds $f(\eps)+g(k)$ eventually.
\end{proof}

\begin{corollary}[growth in the number of servers]\label{cor:separation}
For every fixed $c\in[1,3)$, $R_k^{\R}(c)$ grows exponentially in $k$. At the endpoint,
$R_k^{\R}(3)\ge2k-1$.
\end{corollary}

\begin{proof}
For $c\in(1,3)$, setting $\eps=c-1$ in Theorem~\ref{thm:lower-general} gives the final term
$2(2/(c-1))^{k-1}$, while for $c=1$ monotonicity permits any fixed value in $(1,3)$. At $c=3$,
substituting $\eps=2$ makes every term of the sum equal to $2$ and gives $2k-1$.
\end{proof}

The next two-scale construction gives a finite-$k$ lower bound throughout $c\in(1,3)$.

\begin{proposition}[uniform lower bound for $1<c<3$]\label{prop:curve}
For every $k\ge2$ and $c\in(1,3)$, every $c$-consistent deterministic algorithm for $k$ servers
on the real line has robustness at least $(3+c)/(c-1)$.
\end{proposition}

\begin{proof}
For $k=2$ this is Theorem~\ref{thm:k2} (at $c=1+\eps$, $(3+c)/(c-1)=(4+\eps)/\eps$). For
$k\ge3$, fix $w\in(1,2)$ with $w<4/(1+c)$, and place the servers at $0,1,2,3,\dots,k-1$.

\emph{Padding.} For $\tau=1,\dots,k-3$ the adversary requests the point $\tau+2$ with the
truthful prediction $s_{\tau+2}$. If the algorithm ever matches such a request elsewhere it pays
at least $1$, and the adversary requests every remaining unrequested server point, one per step:
the offline optimum is $0$ while the algorithm has paid $\ge1$, contradicting finite robustness.
So all padding requests are exact hits and the free servers are $\{s_0,s_1,s_2\}$ at
$\{0,1,2\}$.

\emph{Remaining requests.} Request $(1,s_0)$. If the algorithm does not select $s_1$, complete
the input by requesting every remaining server point. This gives positive algorithmic cost and
$\OPT=0$, so the algorithm must select $s_1$. Request
$(w,s_1)$. If the algorithm plays $s_2$, complete with $(2,s_2)$: the full prediction sequence
forms the matching (padding hits, $1\to s_0$, $w\to s_1$, $2\to s_2$) of cost $w$,
which is order-preserving on the sorted requests, hence optimal: perfect, while the algorithm
pays $(2-w)+2$ on the window, ratio $(4-w)/w>c$ by the choice of $w$, inconsistent. So it
plays $s_0$, paying $w$; the adversary finishes with a request at $0$, forced to $s_2$ at cost
$2$. The window requests $\{1,w,0\}$ have $\OPT=2-w$ (order-preserving:
$0\to s_0,\,1\to s_1,\,w\to s_2$), so $R\ge(w+2)/(2-w)$; letting $w\uparrow4/(1+c)$ gives
$(3+c)/(c-1)$.
\end{proof}

\begin{remark}
For $\eps\le1$ the lower bound in Theorem~\ref{thm:lower-general} strictly dominates Proposition~\ref{prop:curve} at every
$k\ge3$, already through the term $8/(c-1)^2$ at $k=3$. Proposition~\ref{prop:curve} also applies
when $\eps>1$.
\end{remark}

\section{The deterministic algorithm}\label{sec:algorithm}

Both parts of the upper bound use the following identity: if $w_t$ is the selected server and
$g_t:=d(r_t,w_t)$ its cost, then the assignment at request $t$ \emph{follows the prediction} when
$w_t=p_t$ and \emph{deviates} otherwise.

\begin{lemma}[cost decomposition]\label{lem:decomp}
On any perfect instance, any algorithm satisfies
$\cost=\OPT+\sum_{t:w_t\ne p_t}\bigl(g_t-d(r_t,p_t)\bigr)$.
\end{lemma}

\begin{proof}
$\cost=\sum_{t:w_t=p_t}d(r_t,p_t)+\sum_{t:w_t\ne p_t}g_t$ and, by perfection,
$\sum_t d(r_t,p_t)=\cost(M^*)=\OPT$; substitute.
\end{proof}

\paragraph{The algorithm.}
For a request $y$, a prediction $p$ and a target $w$, define the \emph{assignment excess}
$\Delta(y;p\to w):=d(y,w)+d(w,p)-d(y,p)\in[0,2d(y,w)]$. Starting with $\sigma_0=0$, the algorithm
$\Sd_\eps$ maintains $\sigma_t\ge0$ after request $t$ and, as long as the predicted labels in the
prefix are distinct, a minimum-cost matching $\Pi_t$ of cost $\Phi_t$ between its free set
$F^A_t$ and the reference free set $F^*_t:=S\setminus\{p_1,\dots,p_t\}$. At step $t$, let
\[
B_t\;:=\;\eps\sum_{j\le t}d(r_j,p_j)\;-\;\sigma_{t-1}.
\]
Fix a total order on server labels. Among minimum-cost matchings, choose $\Pi_t$ to pair every
$u\in F^A_t\cap F^*_t$ to itself, then break ties lexicographically. Such a choice exists because,
by triangle inequality, replacing $(u,q),(v,u)$ by $(u,u),(v,q)$ cannot increase cost.
The algorithm acts as follows while the predicted labels are distinct:
\begin{enumerate}\itemsep2pt
\item Set $b_t=p_t$ if
$p_t\in F^A_{t-1}$. Otherwise, set $b_t=c_t$, the unique $F^A_{t-1}$-side
vertex paired with $p_t$ in $\Pi_{t-1}$.
\item Another free server $w\ne b_t$ is admissible if
$\Delta_t(w):=\Delta(r_t;p_t\to w)\le B_t$. Selecting it sets
$\sigma_t=\sigma_{t-1}+\Delta_t(w)$. Selecting $b_t$ sets $\sigma_t=\sigma_{t-1}$.
\item Select the cheapest admissible server; ties favour $b_t$, then the
fixed server-label order.
\end{enumerate}
If a predicted label repeats, the algorithm selects a nearest free server on this and every
subsequent request, with ties broken by label order, and sets $\sigma_t=\sigma_{t-1}$. Two facts
used below are:
\textup{(F1)}~every free target $w$ has $\Delta_t(w)\le2d(r_t,w)$ (triangle), and hence every
selected deviation has $\Delta_t(w)\le2g_t$;
\textup{(F2)}~by induction, $\sigma_t\le\eps\sum_{j\le t}d(r_j,p_j)$ after every step.

The algorithm is online and deterministic; on a finite explicitly represented metric, each
minimum-cost matching $\Pi_t$ can be recomputed in polynomial time. The construction is specific
to metric matching and uses no auxiliary competitive algorithm.

\begin{theorem}[distinct predictions and consistency]\label{thm:consistency}
For every $k$, every metric space, and every $\eps\in(0,2]$, each input with pairwise distinct
predicted server labels satisfies
\[
\cost(\Sd_\eps)\le(1+\eps)\sum_{t=1}^k d(r_t,p_t).
\]
Consequently, $\Sd_\eps$ is $(1+\eps)$-consistent.
\end{theorem}

\begin{proof}
Fix an input with distinct predicted labels and put $H=\sum_t d(r_t,p_t)$. The matchings $\Pi_t$
are then defined for every prefix, with $\Phi_0=\Phi_k=0$, so it suffices to prove at every step
\[
\bigl(g_t-d(r_t,p_t)\bigr)+\bigl(\Phi_t-\Phi_{t-1}\bigr)\;\le\;
\Delta_t(w_t)\cdot\mathbf 1[w_t\ne b_t],
\]
and then sum and telescope. Four cases; each edits the optimal matching $\Pi_{t-1}$ and uses
only the triangle inequality.
\begin{itemize}\itemsep2pt
\item \emph{The prediction is free and selected.} The excess is $0$, and the identity-pair
convention puts $(p_t,p_t)$ in $\Pi_{t-1}$; deleting that edge leaves a matching of the new free
sets, so $\Phi$ does not increase.
\item \emph{A deviation is selected while $p_t$ is free.} The excess is
$=d(r_t,w)-d(r_t,p_t)=\Delta_t(w)-d(w,p_t)$. Drift: delete $(w,q^*)$ and $(e^A,p_t)$, add
$(e^A,q^*)$. If the two deleted edges coincide, deleting that edge does not increase $\Phi$ and the
claim is immediate; otherwise the rise is at most $d(w,p_t)$, and the sum is at most
$\Delta_t(w)$.
\item \emph{The server $c_t$ is selected.} The excess is
$=d(r_t,c_t)-d(r_t,p_t)\le d(p_t,c_t)$; deleting the edge $(c_t,p_t)$ lowers $\Phi$ by at
least $d(p_t,c_t)$, so the sum is nonpositive.
\item \emph{A deviation $w\ne c_t$ is selected while $p_t$ is unavailable.} The excess is
$\Delta_t(w)-d(w,p_t)$. For the drift,
delete $(c_t,p_t)$ and $(w,q^*)$, add $(c_t,q^*)$: rise at most
$d(c_t,q^*)-d(w,q^*)-d(c_t,p_t)\le d(w,p_t)$. Sum $\le\Delta_t(w)$.
\end{itemize}
Summing and telescoping gives
\[
\cost(\Sd_\eps)-H\le\sigma_k\le\eps H,
\]
where the last inequality is \textup{(F2)}, proving the first statement and, on a correct input
where $H=\OPT$, consistency. The four drift estimates use only distinctness of the predicted
labels and the triangle inequality; before a predicted label repeats, $\Phi$ can rise only on a
deviation, and then by at most $d(w_t,p_t)$.
\end{proof}

\begin{theorem}[robustness in arbitrary metrics]\label{thm:robustness}
For every $k\ge2$, every metric space and every $\eps\in(0,2]$, writing
$\gamma:=2+4/\eps$, $\Sd_\eps$ is
\[
\left[1+\frac4\eps\gamma^{k-2}
 +\frac4\eps\frac{\gamma^{k-2}-1}{\gamma-1}\right]\text{-robust}.
\]
For $\eps\in(0,1]$ this is at most $6^{k-1}/\eps^{k-1}$.
\end{theorem}

\begin{proof}
Fix any instance and an optimal offline matching $o$, with $e_t:=d(r_t,o_t)$ and
$\OPT=\sum_te_t$; let $a_t$ be the algorithm's step costs and $T_t:=\sum_{i\le t}a_i$.

\emph{Free comparison server.} At every step $t$ some free server $v_t$ satisfies
$d(r_t,v_t)\le\bar D_t:=e_t+\sum_{i<t}(a_i+e_i)\le\OPT+T_{t-1}$: follow $r_t\to o_t$; while the
reached server is consumed, say by our step $i$, continue to $r_i\to o_i$. In the union of the
offline matching and our first $t-1$ edges, the component containing $r_t$ is a simple path:
$r_t$ has degree one and every other vertex has degree at most two. Its other endpoint is a server
with no online edge, hence free, and the path length is at most $\bar D_t$.

\emph{Potential function.} While the predicted labels are distinct,
$\Phi_{t-1}\le T_{t-1}+\sum_{i<t}d(r_i,p_i)$: by the pointwise
drift observation at the end of Theorem~\ref{thm:consistency}, $\Phi$ rises only at deviations,
by at most
$d(w_i,p_i)\le a_i+d(r_i,p_i)$.

\emph{Per-step bound.} We prove $a_t\le(2/\eps)\bar D_t+(1+2/\eps)T_{t-1}$. If a predicted label
has repeated, or $v_t$ is the reference server, or $v_t$ is an admissible deviation, then the
selection rule gives $a_t\le d(r_t,v_t)\le\bar D_t$. Otherwise $v_t$ is not admissible:
$\Delta_t(v_t)>B_t$, and since $\Delta_t(v_t)\le2d(r_t,v_t)\le2\bar D_t$ we get $B_t<2\bar D_t$.
By \textup{(F2)} applied through step $t-1$,
$B_t\ge\eps\,d(r_t,p_t)$, so $d(r_t,p_t)<(2/\eps)\bar D_t$. Also, by \textup{(F1)},
$\sigma_{t-1}\le2T_{t-1}$, so
$\eps\sum_{j\le t}d(r_j,p_j)<2\bar D_t+2T_{t-1}$. If $b_t=p_t$, then
$a_t\le d(r_t,p_t)<(2/\eps)\bar D_t$. If $b_t=c_t$, then, writing
$H_t=\sum_{j\le t}d(r_j,p_j)$,
\[
a_t\;\le\;d(r_t,c_t)\;\le\;d(r_t,p_t)+d(p_t,c_t)\;\le\;d(r_t,p_t)+\Phi_{t-1}
\;\le\;T_{t-1}+H_t
\;<\;\frac2\eps\bar D_t+\left(1+\frac2\eps\right)T_{t-1},
\]
using the potential bound and the second consequence.  The identity
$d(r_t,p_t)+H_{t-1}=H_t$ is the cancellation that avoids charging the current prediction twice.
This proves the claimed per-step bound.

\emph{First and terminal steps.} At step $1$, either $o_1=p_1$ or $o_1$ is an admissible
deviation, in which case $a_1\le e_1$, or rejection gives
$\eps d(r_1,p_1)=B_1<\Delta_1(o_1)\le2e_1$ and hence $a_1<(2/\eps)e_1$, so
$T_1=a_1\le(2/\eps)\OPT$ for $\eps\le2$.

After $k-1$ matches, the single free server $u$ is selected by every branch at the terminal step:
it is $b_k$ if the predicted labels are distinct, and the nearest free server otherwise. Since the
$v_k$ must also be $u$,
$a_k\le\bar D_k\le\OPT+T_{k-1}$.

\emph{Recursion.} For $2\le t\le k-1$, the per-step bound and
$\bar D_t\le\OPT+T_{t-1}$ give
$T_t\le\gamma T_{t-1}+(2/\eps)\OPT$, where $\gamma:=2+4/\eps$. Therefore
\[
T_{k-1}\le\left[\frac2\eps\gamma^{k-2}
+\frac2\eps\frac{\gamma^{k-2}-1}{\gamma-1}\right]\OPT,
\]
and the terminal bound gives $\cost\le\OPT+2T_{k-1}$, proving the displayed expression. For
$\eps\le1$, put $q=6/\eps$; since $\gamma\le q$, $q\ge6$, and the geometric sum in the display is
at most $q^{k-2}/(q-1)$, the robustness is at most
$1+(2/3)q^{k-1}+(2/3)q^{k-1}/(q-1)<q^{k-1}$ (with $k=2$ immediate).
\end{proof}

\begin{corollary}[deterministic bounds]\label{cor:bracket}
For every $k\ge2$ and every $\eps\in(0,1]$,
\[
1+\sum_{j=1}^{k-1}\frac{2^{\,j+1}}{\eps^{\,j}}
\le R_k^{\R}(1+\eps)
\le R_k(1+\eps)
\le1+2\left(\frac2\eps+3\right)^{k-1}.
\]
For every $0<\eps\le1/(k-1)$, this implies
\[
\frac{2^k}{\eps^{k-1}}
\le R_k^{\R}(1+\eps)
\le R_k(1+\eps)
\le1+e^{3/2}\frac{2^k}{\eps^{k-1}}.
\]
Consequently, both optimal robustness values are $\Theta_k(1/\eps^{k-1})$ for fixed $k$, and the second display
has constants independent of $k$.
\end{corollary}

\begin{proof}
The lower bound is Theorem~\ref{thm:lower-general}, the middle inequality follows from the definitions,
and the upper bound is Lemma~\ref{lem:closedform}. For the second display, put $q=2/\eps$. Then
\[
(q+3)^{k-1}
=q^{k-1}\left(1+\frac{3\eps}{2}\right)^{k-1}
\le e^{3\eps(k-1)/2}q^{k-1}
\le e^{3/2}q^{k-1}.
\]
Together with the last term $2^k/\eps^{k-1}$ in the lower bound, this completes the proof.
\end{proof}

\begin{remark}
The rejection bound is applied to the first $k-1$ requests; at the final request there is only one
free server, which accounts for the power $k-1$. For $k=2$ the theorem gives $1+4/\eps$, equal to
the value in Theorem~\ref{thm:k2}, while the next section gives a smaller upper bound for $k=3$.
\end{remark}

\section{Three servers}\label{sec:k3}

\begin{theorem}[$k=3$ upper bound]\label{thm:k3ub}
For $k=3$ in every metric and every $\eps\in(0,1]$, $\Sd_\eps$ is
$\bigl(1+12/\eps+8/\eps^2\bigr)$-robust.  For $\eps\in(1,2]$ it is
$\bigl(1+16/\eps+8/\eps^2\bigr)$-robust.  Consequently, for $0<\eps\le1$,
\[
8+4\eps+\eps^2\;\le\;\eps^2R_3^{\R}(1+\eps)\;\le\;8+12\eps+\eps^2,
\qquad\text{so}\qquad
\lim_{\eps\downarrow0}\eps^2R_3^{\R}(1+\eps)=8.
\]
\end{theorem}

\begin{proof}
Fix an input and an optimal offline matching. Let $a_t$ and $e_t$ be the algorithmic and offline
costs at request $t$, let $h_t=d(r_t,p_t)$, and let $\sigma_t$ be the assignment excess charged through
request $t$.  The first-step argument in Theorem~\ref{thm:robustness} gives
\begin{equation}\label{eq:k3a1}
a_1\le\frac2\eps e_1.
\end{equation}
We also have
\begin{equation}\label{eq:k3spent}
\sigma_1\le2e_1.
\end{equation}
The claim is immediate if the algorithm follows $p_1$. If it uses an allowed alternative $w$, let
$o_1$ be the first offline server; if $o_1=b_1$ or is an
admissible alternative, the selection rule gives $a_1\le e_1$, and hence
$\sigma_1=\Delta_1(w)\le2a_1\le2e_1$.  If $o_1$ is rejected, then
\[
\sigma_1=\Delta_1(w)\le B_1<\Delta_1(o_1)\le2e_1.
\]

At request two, the alternating path used in Theorem~\ref{thm:robustness} reaches a free server
$v$ at distance at most
\begin{equation}\label{eq:k3D}
D:=e_2+a_1+e_1.
\end{equation}
If a predicted label has repeated, if $v=b_2$, or if $v$ is an admissible alternative, then
$a_2\le D$.  Otherwise $v$ is rejected, so
\begin{equation}\label{eq:k3H}
\eps(h_1+h_2)-\sigma_1<2D.
\end{equation}

Assume first that $b_2=p_2$. Equations
\eqref{eq:k3spent}--\eqref{eq:k3H} give
\[
a_2\le h_2\le h_1+h_2
<\frac2\eps(D+e_1)
\le\left(\frac4{\eps^2}+\frac4\eps\right)e_1+\frac2\eps e_2.
\]

Suppose instead that $p_2$ was consumed, so $b_2=c_2$. Since $p_1\ne p_2$, request one used
$p_2$ as an alternative to $p_1$. After that move, the two residual
free-server multisets differ only in $p_1,p_2$, so $\Phi_1=d(p_1,p_2)$ and
\[
\sigma_1=a_1+d(p_1,p_2)-h_1.
\]
The server $c_2$ therefore satisfies
\[
a_2\le h_2+d(p_1,p_2)=h_1+h_2+\sigma_1-a_1.
\]
Using \eqref{eq:k3D}, \eqref{eq:k3H}, \eqref{eq:k3a1}, and \eqref{eq:k3spent},
\begin{align*}
a_2
&<\frac2\eps D+\left(1+\frac1\eps\right)\sigma_1-a_1\\
&\le\frac2\eps e_2+\left(\frac2\eps-1\right)a_1
  +\left(2+\frac4\eps\right)e_1\\
&\le\frac2\eps e_2+\left(\frac4{\eps^2}+\frac2\eps+2\right)e_1.
\end{align*}
For $\eps\le1$ this is at most
\begin{equation}\label{eq:k3a2}
a_2\le\frac2\eps e_2+\left(\frac4{\eps^2}+\frac4\eps\right)e_1.
\end{equation}
The easy case $a_2\le D$ also satisfies \eqref{eq:k3a2}.  For $\eps\in(1,2]$, the unsplit bound
\[
a_2\le\frac2\eps e_2+\left(\frac4{\eps^2}+\frac6\eps\right)e_1
\]
follows directly from \eqref{eq:k3spent}--\eqref{eq:k3H} and the potential estimate
$\Phi_1\le a_1+h_1$.

At request three there is one free server, and the terminal alternating path gives
\[
a_3\le e_3+(a_1+e_1)+(a_2+e_2),
\]
so $a_1+a_2+a_3\le\OPT+2a_1+2a_2$. Substituting
\eqref{eq:k3a1} and \eqref{eq:k3a2} proves the first upper bound, while the unsplit estimate proves
the second. Corollary~\ref{cor:values} then gives
\[
8+4\eps+\eps^2\le\eps^2R_3^{\R}(1+\eps)
\le8+12\eps+\eps^2
\]
for $0<\eps\le1$, which proves the limit.
\end{proof}

On the real line the second-request estimate can be sharpened, and once the slack is small
enough the sharpened form closes the gap between Theorem~\ref{thm:k3ub} and the lower bound of
Corollary~\ref{cor:values}, determining the three-server value exactly.

\begin{theorem}[exact three-server value for small slack]\label{thm:k3exact}
Let $\eps^\star:=\sqrt{13}-3=0.6055\ldots$, equivalently the value of $\eps$ at which $q=2/\eps$
solves $q^2=3q+1$.  For $k=3$ on the real line and every $\eps\in(0,\eps^\star]$,
\[
\cost(\Sd_\eps)\le\left(1+\frac4\eps+\frac8{\eps^2}\right)\OPT
\qquad\text{on every input,}
\]
and consequently
\[
R_3^{\R}(1+\eps)=1+\frac4\eps+\frac8{\eps^2}
\qquad(0<\eps\le\eps^\star).
\]
\end{theorem}

\begin{proof}
Retain the notation of the proof of Theorem~\ref{thm:k3ub} and put $q=2/\eps$, so that
$q\ge2/\eps^\star=(3+\sqrt{13})/2>3$ and in particular $q\ge2$.  We prove the sharpened
second-request estimate
\begin{equation}\label{eq:k3exacta2}
a_2\;\le\;q^2e_1+qe_2\;=\;\frac4{\eps^2}e_1+\frac2\eps e_2
\end{equation}
in every branch, and then sum.

\emph{Two facts about the first request.}  Cheapest-allowed selection and the fact that the
reference server $b_1=p_1$ is always allowed give $a_1\le h_1$.  On the real line an
alternative server no farther from the request than the predicted server therefore either lies
between the two, giving assignment excess zero, or lies on the far side of the request, giving
excess exactly $2a_1$; the case beyond $p_1$ is excluded by $a_1\le h_1$.  Hence
\begin{equation}\label{eq:k3dichotomy}
\sigma_1\in\{0,2a_1\}.
\end{equation}
When $\sigma_1=2a_1$,
\begin{equation}\label{eq:k3a1e1}
a_1\le e_1 .
\end{equation}
Indeed, if the offline server $o_1$ equals $b_1$ or is an admissible alternative, then
cheapest-allowed selection gives $a_1\le d(r_1,o_1)=e_1$; and if $o_1$ is rejected, then
allowance of the move the algorithm did make together with that rejection gives
$2a_1=\sigma_1\le B_1=\eps h_1$ and $\eps h_1=B_1<\Delta_1(o_1)\le2e_1$, so $a_1<e_1$.
When $\sigma_1=0$ the algorithm either followed its prediction, so $a_1=h_1$, or used a
zero-excess alternative $w_1$, so that $h_1=a_1+d(w_1,p_1)\ge a_1$; in both cases
\begin{equation}\label{eq:k3h1a1}
h_1\ge a_1 .
\end{equation}

\emph{Branch 1: $a_2\le D$.}  By \eqref{eq:k3D} and \eqref{eq:k3a1},
$a_2\le e_2+a_1+e_1\le e_2+(q+1)e_1$, which is at most $qe_2+q^2e_1$ because $q\ge2$.

\emph{Branch 2: $v$ is rejected and the reference server at request two is the free
prediction.}  Then $a_2\le h_2$ and \eqref{eq:k3D}--\eqref{eq:k3H} give
\[
a_2<\frac2\eps e_2+
\left[\frac2\eps(a_1+e_1)+\frac{\sigma_1}\eps-h_1\right]
=qe_2+\Bigl[q(a_1+e_1)+\tfrac q2\sigma_1-h_1\Bigr].
\]
If $\sigma_1=0$, then \eqref{eq:k3h1a1} and \eqref{eq:k3a1} bound the bracket by
$(q-1)a_1+qe_1\le(q-1)qe_1+qe_1=q^2e_1$.  If $\sigma_1=2a_1$, then allowance gives
$2a_1\le\eps h_1$, that is $h_1\ge qa_1$, so the bracket is at most
$2qa_1+qe_1-qa_1=q(a_1+e_1)\le2qe_1\le q^2e_1$ by \eqref{eq:k3a1e1} and $q\ge2$.

\emph{Branch 3: $v$ is rejected and the reference server at request two is paired with a
consumed prediction.}  Since $k=3$ and only request one has consumed a server, request one used
$p_2$, the two residual free-server multisets differ only in $p_1,p_2$, and $\Phi_1=d(p_1,p_2)$.
As in the derivation of \eqref{eq:k3a2}, $a_2\le h_2+\Phi_1=H_2+\sigma_1-a_1$, so
\eqref{eq:k3D}--\eqref{eq:k3H} give
\[
a_2<\frac2\eps e_2+
\left(\frac2\eps-1\right)a_1+
\frac2\eps e_1+
\left(1+\frac1\eps\right)\sigma_1
=qe_2+(q-1)a_1+qe_1+\Bigl(1+\tfrac q2\Bigr)\sigma_1 .
\]
If $\sigma_1=0$, then \eqref{eq:k3a1} gives $a_2<qe_2+(q-1)qe_1+qe_1=qe_2+q^2e_1$.  If
$\sigma_1=2a_1$, then $\bigl(1+\tfrac q2\bigr)\sigma_1=(q+2)a_1$ and \eqref{eq:k3a1e1} give
\[
a_2<qe_2+qe_1+(2q+1)a_1\le qe_2+(3q+1)e_1\le qe_2+q^2e_1,
\]
the last step being exactly the inequality $q^2\ge3q+1$, that is $\eps\le\eps^\star$.  This is
the only point at which the restriction $\eps\le\eps^\star$ is used.

\emph{Summation.}  The terminal alternating path gives $a_3\le e_3+(a_1+e_1)+(a_2+e_2)$, so
$\cost(\Sd_\eps)\le\OPT+2a_1+2a_2$.  By \eqref{eq:k3a1} and \eqref{eq:k3exacta2},
\[
\cost(\Sd_\eps)\le\OPT+2qe_1+2\bigl(q^2e_1+qe_2\bigr)
\le\OPT+2q(q+1)(e_1+e_2)\le\bigl(1+2q+2q^2\bigr)\OPT,
\]
which is the displayed bound.  Corollary~\ref{cor:values} supplies the matching lower bound.
\end{proof}

Above $\eps^\star$ only Branch 3 with $\sigma_1=2a_1$ escapes the argument, and there a
geometric analysis recovers the same bound in part of the remaining range.

\begin{proposition}[the residual configuration on the real line]\label{prop:k3residual}
Let $k=3$ on the real line and $\eps^\star<\eps\le1$.  Then
$\cost(\Sd_\eps)\le\bigl(1+4/\eps+8/\eps^2\bigr)\OPT$ on every input except possibly those in
which all of the following hold: request one uses $W=p_2$ as a positive-excess alternative to
$P=p_1$; at request two the reference server is $P$, which is paired with the consumed
prediction $W$; the other free server $U$ is rejected at request two; and the order-preserving
offline optimum does \emph{not} match $(r_1,r_2,r_3)$ to $(U,W,P)$ respectively.
\end{proposition}

\begin{proof}
By the proof of Theorem~\ref{thm:k3exact}, \eqref{eq:k3exacta2} holds in Branches 1 and 2 and in
Branch 3 with $\sigma_1=0$, for every $\eps\le1$; the summation step then applies verbatim.  The
excluded configuration is Branch 3 with $\sigma_1=2a_1$, and it remains to treat it when the
offline optimum matches $r_1$ to $U$, $r_2$ to $W$ and $r_3$ to $P$.  Fix such an optimum, put
$q=2/\eps\ge2$, and translate
and, if necessary, reflect the line so that
\[
r_1=0,\qquad W=-a,\qquad P=h,
\]
where $a,h>0$. Since the first excess is $2a$, allowance gives $h\ge qa$. Write
\[
U=z,\qquad e=|z|=e_1,\qquad f=d(r_2,W)=e_2,\qquad g=d(r_3,P)=e_3.
\]
The two remaining algorithmic costs satisfy
\[
d(r_2,P)\le f+h+a,\qquad d(r_3,U)\le g+d(P,U).
\]
Consequently,
\begin{equation}\label{eq:k3terminalgeometry}
\cost(\Sd_\eps)\le\OPT+X,
\qquad X:=2a+h+d(P,U)-e.
\end{equation}

First suppose $0\le z\le h$. The server $U$ has zero assignment excess and is therefore allowed, so the
selection of $W$ gives $e=z\ge a$; order preservation gives $r_2\le0$, while
$d(P,U)=h-e$ and $X=2(a+h-e)$. If $r_2\le-a$, the assignment excess at request two for $U$ is $2(e+a)$, and
rejection of $U$ gives
\[
h<qe+2qa-f,
\]
and hence
\[
\frac X2<(2q+1)a+(q-1)e-f\le3qe-f\le q(q+1)(e+f).
\]
If instead $-a\le r_2\le0$, the assignment excess is $2(e+a-f)$ and rejection gives
\[
h<qe+2qa-(q+1)f.
\]
Therefore
\[
\frac X2<(2q+1)a+(q-1)e-(q+1)f
\le3qe-(q+1)f\le q(q+1)(e+f).
\]
Both final inequalities use $q\ge2$.

If $z\ge h$, then $d(P,U)=e-h$ and $X=2a$, while $e\ge h\ge qa$ gives the required bound.
Finally suppose $z\le-a$, and write $z=-e$, where $e\ge a$. Order preservation gives
$r_2\ge0$, and the assignment excess at request two for $U$ is $2(e-a)$, so rejection reduces to
\[
h+f<qe.
\]
Since $X=2(a+h)$, we obtain
\[
\frac X2<a+qe-f\le(q+1)e-f\le q(q+1)(e+f).
\]
Thus $X\le2q(q+1)\OPT$ whenever the offline optimum matches $r_1$ to $U$, and
\eqref{eq:k3terminalgeometry} gives
\[
\cost(\Sd_\eps)\le
\bigl(1+2q+2q^2\bigr)\OPT
=\left(1+\frac4\eps+\frac8{\eps^2}\right)\OPT
\]
in that case as well.
\end{proof}

\begin{remark}[what is left open at $k=3$]\label{rem:k3open}
For $\eps\le\eps^\star$ the three-server value is settled by Theorem~\ref{thm:k3exact}, and the
argument is purely arithmetic. For $\eps^\star<\eps\le1$, a single branch of the second-request
analysis remains open. Proposition~\ref{prop:k3residual} settles that branch when the
order-preserving offline optimum matches $(r_1,r_2,r_3)$ to $(U,W,P)$ respectively; the other
assignments of the three requests to the three servers are open. Once $q^2<3q+1$, the per-step identities and \eqref{eq:k3a1},
\eqref{eq:k3a1e1} no longer imply \eqref{eq:k3exacta2}, so the remaining assignments require the
geometry of the line. We expect the value to be $1+4/\eps+8/\eps^2$ throughout $0<\eps\le1$.
Remark~\ref{rem:k3endpoint} shows that $\Sd_\eps$ itself does not attain this expression for
$\eps>1$, leaving the value above $\eps=1$ open.
\end{remark}

\begin{conjecture}\label{conj:k3}
$R_3^{\R}(1+\eps)=1+4/\eps+8/\eps^2$ for every $\eps\in(0,1]$.
\end{conjecture}

\begin{remark}[the endpoint for this algorithm]\label{rem:k3endpoint}
The restriction $\eps\le1$ is necessary for $\Sd_\eps$ to attain the expression of
Theorem~\ref{thm:k3exact}.
For $1<\eps\le2$, take servers $W=-a$, $U=a$, and $P=b$, requests $0,W,P$, and predictions
$P,W,U$, give $W$ the earlier alternative label, and let $b$ approach $6a/\eps$ from below. The
algorithm uses $W$ at the first request, then rejects $U$ because $\eps b-2a<4a$ and uses $P$ at
the second, leaving $U$ for the final match. Its ratio approaches $1+12/\eps$, which exceeds
$1+4/\eps+8/\eps^2$ throughout this interval and leaves the exact value above $\eps=1$ open.
\end{remark}

\begin{lemma}[two prefix bounds]\label{lem:fixedkprefix}
Let $\delta_t=\sigma_t-\sigma_{t-1}$. Fix an arbitrary input and
an optimal offline matching, and write
\[
a_t=d(r_t,w_t),
\qquad
e_t=d(r_t,o_t),
\qquad
h_t=d(r_t,p_t),
\]
where $w_t$ is the server selected by $\Sd_\eps$ and $o_t$ is the offline server. Put
\[
T_t=\sum_{i\le t}a_i,
\qquad
H_t=\sum_{i\le t}h_i,
\qquad
D_t=e_t+\sum_{i<t}(a_i+e_i).
\]
If the first $t$ predictions are distinct, then
\begin{equation}\label{eq:fixedkpotential}
T_t-H_t+\Phi_t\le\sigma_t.
\end{equation}
For every request, including requests after a repeated prediction,
\begin{equation}\label{eq:fixed-excess}
\delta_t\le2D_t.
\end{equation}
\end{lemma}

\begin{proof}
Before a repeated prediction, the residual-matching update in the proof of
Theorem~\ref{thm:consistency} gives
\[
(a_t-h_t)+(\Phi_t-\Phi_{t-1})\le\delta_t.
\]
If $p_t$ is free, selecting it gives zero service excess and deletes its identity edge; selecting
a permitted target $w_t$ instead reconnects the two affected residual-matching edges, increasing
the potential by at most $d(w_t,p_t)$, while
\[
a_t-h_t=\Delta_t(w_t)-d(w_t,p_t).
\]
If $p_t$ has been consumed, selecting $c_t$ deletes the edge $(c_t,p_t)$ and decreases the
potential by at least the service excess, while selecting another permitted target uses the same
two-edge reconnection. Summation proves \eqref{eq:fixedkpotential}.

For \eqref{eq:fixed-excess}, the alternating path in the robustness proof gives a free server $v_t$
with $d(r_t,v_t)\le D_t$. The case $\delta_t=0$ is immediate; otherwise the algorithm uses a
permitted deviation, and if $v_t=b_t$ or $v_t$ is permitted, least-cost selection gives
$a_t\le d(r_t,v_t)$, and hence $\delta_t\le2a_t\le2D_t$. If $v_t$ is rejected, then
\[
\delta_t\le B_t<\Delta_t(v_t)\le2d(r_t,v_t)\le2D_t.
\]
\end{proof}

\begin{theorem}[leading coefficient for fixed $k$]\label{thm:fixedkleading}
Let $q=2/\eps$, and define polynomials
\[
P_1(q)=q,
\qquad
Q_1(q)=2,
\]
and, for $t\ge2$,
\[
\begin{aligned}
P_t(q)
&=(q+1)P_{t-1}(q)
 +\left(1+\frac q2\right)Q_{t-1}(q)+q,\\
Q_t(q)
&=Q_{t-1}(q)+2P_{t-1}(q)+2.
\end{aligned}
\]
For every $k\ge2$ and every $\eps\in(0,1]$, in every metric space,
\[
\cost(\Sd_\eps)
\le\bigl(1+2P_{k-1}(2/\eps)\bigr)\OPT.
\]
Consequently, for every fixed $k\ge2$,
\[
\lim_{\eps\downarrow0}\eps^{k-1}R_k^{\R}(1+\eps)
=\lim_{\eps\downarrow0}\eps^{k-1}R_k(1+\eps)=2^k.
\]
\end{theorem}

\begin{proof}
Fix the input and offline matching used in Lemma~\ref{lem:fixedkprefix}. We prove, for every
nonterminal request $t\le k-1$,
\begin{equation}\label{eq:fixedkrecurrencebounds}
T_t\le P_t(q)\OPT,
\qquad
\sigma_t\le Q_t(q)\OPT.
\end{equation}
At request one, the rejection argument in Theorem~\ref{thm:robustness} gives
$T_1=a_1\le q e_1\le q\OPT$, while \eqref{eq:fixed-excess} gives
$\sigma_1\le2e_1\le2\OPT$.

Fix $t\in\{2,\ldots,k-1\}$ and suppose \eqref{eq:fixedkrecurrencebounds} holds at $t-1$. Since the
$e_i$ are the edges of one offline optimum,
\begin{equation}\label{eq:fixedkD}
D_t=e_t+T_{t-1}+\sum_{i<t}e_i
\le T_{t-1}+\OPT.
\end{equation}
If a prediction has repeated, the algorithm uses the nearest free server; if the predictions are
distinct and $v_t=b_t$ or is permitted, least-cost selection gives the same estimate. In both
cases $a_t\le D_t$, and the claimed recurrence below follows because $q\ge2$.

It remains to consider a distinct prediction prefix for which $v_t$ is rejected. The admissibility test
and $\Delta_t(v_t)\le2D_t$ give
\[
H_t<qD_t+\frac q2\sigma_{t-1}.
\]
If $p_t$ is free then $T_t\le T_{t-1}+H_t$; if it has been consumed, the edge $(c_t,p_t)$ has
length at most $\Phi_{t-1}$ and \eqref{eq:fixedkpotential} gives
\[
T_t\le H_t+\sigma_{t-1}.
\]
Using \eqref{eq:fixedkD}, every branch therefore satisfies
\[
T_t\le
(q+1)T_{t-1}
+\left(1+\frac q2\right)\sigma_{t-1}
+q\OPT.
\]
The first polynomial recurrence proves the first inequality in
\eqref{eq:fixedkrecurrencebounds}, while Equations~\eqref{eq:fixed-excess} and~\eqref{eq:fixedkD} give
\[
\sigma_t\le\sigma_{t-1}+2T_{t-1}+2\OPT,
\]
and the second recurrence proves the remaining induction step.

At request $k$, the terminal alternating path to the single free server gives
\[
a_k\le\OPT+T_{k-1}.
\]
and hence
\[
\cost(\Sd_\eps)
\le\OPT+2T_{k-1}
\le\bigl(1+2P_{k-1}(q)\bigr)\OPT.
\]
which also covers $\OPT=0$.

Induction in the polynomial recurrences gives
\[
P_t(q)=q^t+O_t(q^{t-1}),
\qquad
Q_t(q)=2q^{t-1}+O_t(q^{t-2}).
\]
The upper bound therefore gives a limiting coefficient at most $2^k$ for both $R_k$ and
$R_k^{\R}$. The highest-order term $2^k/\eps^{k-1}$ in Theorem~\ref{thm:lower-general} gives the
reverse inequality for $R_k^{\R}$ and, because
$R_k^{\R}\le R_k$, for $R_k$ as well.
\end{proof}

\begin{lemma}[closed form of the recurrence]\label{lem:closedform}
Let $q\ge2$ and $\mu=q+3$. For every $t\ge1$,
\[
P_t(q)\le\mu^t.
\]
Consequently, for every metric space and every $0<\eps\le1$,
\[
\cost(\Sd_\eps)
\le\left(1+2\left(\frac2\eps+3\right)^{k-1}\right)\OPT.
\]
\end{lemma}

\begin{proof}
Put
\[
\beta=\frac{2(q+4)}{(q+2)(q+3)}.
\]
We prove simultaneously that $P_t(q)\le\mu^t$ and $Q_t(q)\le\beta\mu^t$, using the identities
\[
\left(1+\frac q2\right)\beta=1+\frac1\mu,
\qquad
\beta(q+2)=2+\frac2\mu
\].
At $t=1$, one has $P_1=q\le\mu$ and
$Q_1=2\le\beta\mu=2(q+4)/(q+2)$.

Suppose both bounds hold at $t-1$. The first recurrence gives
\[
P_t
\le\left(\mu-1+\frac1\mu\right)\mu^{t-1}+q
=\mu^t-\mu^{t-1}+\mu^{t-2}+q
\le\mu^t,
\]
because $\mu^{t-1}-\mu^{t-2}=\mu^{t-2}(q+2)>q$. For the second recurrence,
\[
\beta\mu^t-Q_t
\ge\beta\mu^t-(\beta+2)\mu^{t-1}-2
=2(\mu^{t-2}-1)\ge0.
\]
This proves the simultaneous induction, and the cost bound follows from
Theorem~\ref{thm:fixedkleading} after setting $q=2/\eps$.
\end{proof}

\begin{remark}[lower-bound comparison]
The three-server construction in Theorem~\ref{thm:lower-general} uses servers $0,1,a$ with $a$ just
below $x+2/\eps$.  Every $(1+\eps)$-consistent algorithm follows two predictions, including one
from the consumed middle-server location, and incurs ratio
$(2a-x)/x\to1+4(2+\eps)/\eps^2$.  The predicted cost already incurred before the second decision
widens the admissible interval to $a\ge x+2/\eps$.  The lower construction places $a$ immediately
below this boundary.
\end{remark}

\section{Randomisation}\label{sec:randomised}

For a randomised algorithm, a perfect completion bounds the probability that the algorithm does
not follow a prediction, which changes the optimal two-server constant.

\begin{theorem}[randomised two-server trade-off]\label{thm:randk2}
For $\eps\in(0,1]$, $R_2^{\mathrm{rand},\R}(1+\eps)=1+1/\eps$; for $\eps\in[1,2]$ the value is $2$.
The two formulas agree at $\eps=1$.
\end{theorem}

\begin{proof}
Normalise the servers to $0,1$, take $p_1=1$, $x=r_1\in(0,1)$, put $z=(1-x)/x$, and let $q$ be
the probability of following the prediction. The two matchings have costs
$c_{\rm id}=x+|y-1|$ (select $s_0$ first) and $c_{\rm sw}=1-x+|y|$ (select $s_1$ first),
and~(\ref{eq:phipsi}) says their
worst relative excesses are $2/z$ and $2z$, respectively.

For the lower bound when $\eps\le1$, choose $z=1/\eps$. Expected consistency on the perfect
completion $y=0$ requires $(1-q)2x\le\eps(1-x)$ and hence $q\ge1/2$, while on the completion
$y=1$ we have $\OPT=x$ and expected ratio $1+2qz\ge1+1/\eps$.

For the upper bound, follow the prediction with probability
$q_\eps(z)=\max\{1/(1+z^2),1-\eps z/2\}$. (Outside $(0,1)$ take the move that is weakly
optimal for every completion; reflect for $p_1=0$; at the forced second step ignore its
prediction and use the unique free server.) On a perfect instance the expected ratio is at most
$1+2(1-q_\eps)/z\le1+\eps$. If following the prediction is optimal, the same expression is at most
$1+2z/(1+z^2)\le2$. If the other matching is optimal, the ratio is at most $1+2q_\eps z$: according to
which term defines the maximum, this is at most either $1+2z/(1+z^2)\le2$ or
$1+2z-\eps z^2\le1+1/\eps$. This proves the first claim. For $\eps\ge1$, the balanced choice
$q=1/(1+z^2)$ is $2$-consistent and $2$-robust. At $x=1/2$, the completions $y=0,1$ have
expected ratios $1+2(1-q)$ and $1+2q$, so every randomised algorithm has robustness at least $2$.
\end{proof}

For general $k$, a single perfect completion bounds all first deviations at one scale. Suppose the
main prefix through stage $j$ is
$(x,s_1),(a_1,s_2),\dots,(a_{j-1},s_j)$. Complete it by $(0,s_0)$ and then every later exact hit
$(a_i,s_i)$, $i>j$; this instance is perfect, with $\OPT=a_j-x$. If $E_j$ is the event that the
first deviation occurs at stage $j$, the remaining-cost computations from
Theorem~\ref{thm:lower-general} give
\begin{equation}\label{eq:randcharge}
\Pr(E_j)g_j\le\eps(a_j-x),\qquad
g_j=\begin{cases}
2x,&j=1,\ k=2,\\
\min\{2x,2(a_2-a_1)\},&j=1,\ k\ge3,\\
\min\{2a_{j-1},2(a_{j+1}-a_j)\},&2\le j\le k-2,\\
2a_{k-2},&j=k-1.
\end{cases}
\end{equation}
Every completed matching costs at least $\OPT$; on $E_j$, selecting a server to the left adds the
first quantity in the minimum, while selecting $s_m$ to the right adds $2(a_m-a_j)$. Expected
consistency therefore implies~(\ref{eq:randcharge}) using a completion fixed independently of the
random bits and valid for every target simultaneously.

\begin{theorem}[fixed-$k$ randomised lower bound]\label{thm:randk}
For every $k\ge2$ and $\eps\le1/(8k)$,
\[
R_k^{\mathrm{rand},\R}(1+\eps)\ge
\frac{5}{16(4k)^{k-2}}\frac1{\eps^{k-1}}.
\]
Consequently $R_k^{\mathrm{rand},\R}(1+\eps)=\Theta_k(1/\eps^{k-1})$ as $\eps\downarrow0$.
\end{theorem}

\begin{proof}
Set $x=2\eps$, $\lambda=1/(4k\eps)\ge2$, $a_1=1$, and $a_j=\lambda^{j-1}$, then use the
construction from Theorem~\ref{thm:lower-general} and finish with $(r_k,p_k)=(a_{k-1},s_0)$. Equation~(\ref{eq:randcharge})
gives $\Pr(E_1)\le(1-x)/4\le1/4$, and since
$a_{j+1}-a_j=a_{j-1}\lambda(\lambda-1)\ge2a_{j-1}$, it gives
$\Pr(E_j)<\eps\lambda/2=1/(8k)$ for every $2\le j\le k-1$. Since the first-deviation events are
disjoint, the probability $p$ of following all $k-1$ predictions is greater than $5/8$.

The terminal instance has $\OPT=x$; every realisation costs at least $x$, and when all predictions
are followed the cost is $2a_{k-1}-x$. Hence
\[
\frac{\mathbb E[\cost]}{\OPT}
\ge1+2p\left(\frac{a_{k-1}}x-1\right)
>1+\frac54\left(\frac{a_{k-1}}x-1\right).
\]
Here $a_{k-1}/x\ge8$, so the last expression is at least
$(5/8)a_{k-1}/x=(5/16)\lambda^{k-2}/\eps$, the claimed bound. The matching upper bound follows
from Theorem~\ref{thm:robustness}, since a deterministic algorithm is also randomised.
\end{proof}

\subsection{Uniform randomised bounds}\label{sec:randfixedband}

The preceding theorem determines the power of $1/\eps$ for fixed $k$, but its explicit
probability bound requires $\eps=O(1/k)$. The next results give bounds uniform in $k$: the lower
bound is exponential in $1/\eps$, and the upper bound combines a prediction-following algorithm
with a prediction-free competitive algorithm.

\begin{theorem}[uniform randomised lower bound]\label{thm:randuniformlb}
For every $\eta>0$ there are constants $\eps_0(\eta)>0$ and $C_\eta<\infty$ such that, for every
$\eps\in(0,\eps_0(\eta)]$ and every $k\ge C_\eta/\eps$,
\[
R_k^{\mathrm{rand},\R}(1+\eps)\ge \exp\!\left(\frac{2-\eta}{\eps}\right).
\]
\end{theorem}

\begin{proof}[Proof sketch]
Fix $x=1/2$, put $a_1=1$, and take
$a_i=(1+\delta)^{i-1}$, where $\delta>0$ is a sufficiently small constant depending only on
$\eta$. Use the common lower-bound prefix
$(x,s_1),(a_1,s_2),\ldots,(a_{n-1},s_n)$, its perfect stopping instances, and the terminal completion
at $a_n$.

For a deterministic realisation, let $i$ be the first stage at which it consumes the leftmost
server $s_0=0$, with $i=\infty$ if this never happens.  On the real line, for any perfect matching,
$\cost=\Delta+2N$, where $\Delta=\sum(s-r)$ depends only on the two point multisets and $N$ is the
total leftward displacement (Lemma~\ref{lem:linecost}). Thus, if the realisation uses the edge
$b_i\to0$, where $b_1=x$ and $b_i=a_{i-1}$ for $i\ge2$, its cost on every later stopping instance is at
least $\OPT+2b_i$, and its cost on the terminal instance is at least $2b_i-x$. If $s_0$ is never
used in the prefix, the final request is forced to $s_0$ and the total cost is at least $2a_n-x$.
Averaging over the random tape shows that every randomised algorithm is coordinatewise dominated
on this family by a mixture that records only the first stage at which $s_0$ is used
(Theorem~\ref{thm:canon}).

Writing $q_i$ for this probability, $(1+\eps)$-consistency on the stopping instance at depth $j$
therefore gives the exact prefix constraints
\[
2\sum_{i\le j}q_i b_i\le \eps(a_j-x).
\]
The expected cost on the terminal instance is minimised by assigning probability mass to the
earliest scales (Theorem~\ref{thm:greedy}). Here the first capacity is $\eps/2$, and each subsequent
geometric scale contributes capacity $\eps\delta/2$. Hence roughly $2/(\eps\delta)$ stages are
required before one unit of probability is accumulated. Retaining the fully filled coordinates
gives cost at least
\[
\eps\bigl((1+\delta)^L-1\bigr)-\frac12,
\qquad
L=\left\lfloor\frac{2(1-\eps/2)}{\eps\delta}\right\rfloor.
\]
Since $2\log(1+\delta)/\delta\ge2-\eta/3$ for sufficiently small $\delta$, division by
$\OPT=1/2$ yields the displayed bound for all sufficiently small $\eps$.  The construction uses
$O_\eta(1/\eps)$ servers; exact-hit padding transfers it to every larger $k$. The three asserted
steps, the cost decomposition on the real line, the reduction to this mixture, and the
greedy minimisation subject to the prefix constraints, are proved in Appendix~\ref{app:lb}.
\end{proof}

The upper bound follows from the following comparison theorem.

\begin{definition}[switching assumptions]\label{def:switching}
Fix an online problem and two deterministic online algorithms $X,Y$, with per-request costs
$x_t,y_t$ and cumulative costs $X_t,Y_t$. The pair satisfies the \emph{switching assumptions} if there is one
uniform online implementation, common to all inputs and all mode schedules, that starts with the
algorithm and both references in the same state, with tracking potential $\Psi_0=0$, maintains a
mode in $\{X,Y\}$ and a nonnegative potential $\Psi$, and supports the following operations
pathwise:
\begin{enumerate}
\item[(S)] while tracking one reference on a request, actual cost plus the change in $\Psi$ is at most
that reference's cost on the request;
\item[(T)] changing the tracked reference after prefix $t$, without changing the online
state, increases $\Psi$ by at most $D_t\le X_t+Y_t$.
\end{enumerate}
The two deterministic reference trajectories are simulated in parallel, so their current actions
and costs are known after the request is revealed and before the combined algorithm acts. The
implementation must remain valid after every finite sequence of pre-service and post-service
switches.  The switching kernels are measurable; Bernoulli maximal couplings below give an
explicit implementation.
\end{definition}

Fix $L>0$, let $Z=e^L-1$, and define
\[
F(z)=\min\left\{1,\frac{\log(1+z)}L\right\}.
\]
Before request $t$, put $A=X_{t-1}$, $B=Y_{t-1}$, and let $x,y$ be the two reference costs on the
current request. Along the virtual interpolation
\[
A(s)=A+sx,\qquad B(s)=B+sy,\qquad q(s)=F(A(s)/B(s)),
\]
let $\bar q=\int_0^1q(s)\,ds$. The combined algorithm first changes the marginal probability of tracking
$Y$ from $q(0)$ to $\bar q$, serves the request in the resulting mode, and then changes the
marginal from $\bar q$ to $q(1)$. Each change uses a maximal Bernoulli coupling. Since
$A(s)/B(s)$ is monotone on the request interval, Definition~\ref{def:switching} gives the
expected amortised line integral
\begin{equation}\label{eq:interpolation}
\int\bigl[(1-F(z))\,dX+F(z)\,dY+(X+Y)|dF(z)|\bigr].
\end{equation}

\begin{theorem}[comparison theorem]\label{thm:comparison}
Let $c\ge1$ and $\eps\in(0,1]$, and put
\[
\kappa(c)=2\log c+(c+1)\log(1+1/c),\qquad
L=\frac{2+\kappa(c)}\eps,\qquad
M=1+\frac{2(e^L-1)}L.
\]
For every pair $(X,Y)$ satisfying Definition~\ref{def:switching}, there is a randomised online algorithm such that, on
every fixed input,
\[
\mathbb E[\cost]\le M\,\cost(Y),
\]
and, on every input satisfying $\cost(Y)\le c\,\cost(X)$,
\[
\mathbb E[\cost]\le(1+\eps)\cost(X).
\]
\end{theorem}

\begin{proof}[Proof sketch]
Using $(X+Y)|dz|\le(1+z)dX+z(1+z)dY$, the integrand in
\eqref{eq:interpolation} is at most $a(z)dX+b(z)dY$, where
\[
a(z)=1-F(z)+(1+z)F'(z),\qquad
b(z)=F(z)+z(1+z)F'(z).
\]
The proof uses two potential functions. For the comparison with $X$, define
below the cap
\[
\phi(z)=\frac{2z\log z-(z+1)\log(1+z)+\kappa(c)z}{L}
\]
and extend it linearly above $Z$.  The potential $V(X,Y)=Y\phi(X/Y)$ cancels the $dY$ term and
gives
\[
d(\text{cost})+dV\le\left(1+\frac{2+\kappa(c)}L\right)dX.
\]
The choice of $\kappa(c)$ makes $V$ nonnegative whenever $Y\le cX$.  For the unconditional
comparison to $Y$, set $\chi(z)=\int_z^Za(s)\,ds$ below the cap and $0$ above it.  Then
$W(X,Y)=Y\chi(X/Y)$ cancels the $dX$ term and gives
\[
d(\text{cost})+dW\le\left(1+\frac{2(e^L-1)}L\right)dY.
\]
The full derivation, including zero-cost prefixes and the discrete implementation with a switch
before and after service,
appears in Appendix~\ref{app:comparison}.
\end{proof}

The next choice of $F$ gives the constant in the exponential dependence on $1/\eps$. Its full
proof is in Appendix~\ref{app:shifted}.

\begin{theorem}[comparison using a shifted logarithm]\label{thm:shifted-comparison}
Let $c\ge1$ and $\eps\in(0,1]$, and define
\[
K=2+2\eps\log\frac c\eps+(c+2\eps-1)\log\left(1+\frac\eps c\right),
\quad L=\frac K\eps,\quad Z=\frac{e^L-1}{\eps},
\]
and
\[
\widehat M(c,\eps)=1+\frac{2\eps Z(1+Z)}{Le^L}.
\]
For every pair $(X,Y)$ satisfying Definition~\ref{def:switching}, there is a randomised online algorithm satisfying
$\mathbb E[\cost]\le\widehat M(c,\eps)\cost(Y)$ on every fixed input and
$\mathbb E[\cost]\le(1+\eps)\cost(X)$ whenever $\cost(Y)\le c\cost(X)$.  For every fixed
$c\ge1$,
\[
\log\widehat M(c,\eps)=\frac2\eps+O_c(\log(1/\eps))\qquad(\eps\downarrow0).
\]
The online mixing rule is the preceding Bernoulli construction with a switch before and after
service, and with
$F(z)=\min\{1,\log(1+\eps z)/L\}$.
\end{theorem}

A different choice of $F$ gives simultaneous comparisons without an assumption relating the
terminal costs of $X$ and $Y$. The proof is in Appendix~\ref{app:truncated}.

\begin{theorem}[comparison using a truncated logarithm]\label{thm:truncated-comparison}
Let $\eps\in(0,1]$ and $r>0$, and define
\[
L=\frac{2(1+1/r)}\eps,
\qquad
Z=re^L.
\]
For every pair $(X,Y)$ satisfying Definition~\ref{def:switching}, there is a randomised online algorithm satisfying, on
every fixed input,
\[
\mathbb E[\cost]\le(1+\eps)\cost(X)
\]
and
\[
\mathbb E[\cost]\le M_r(\eps)\cost(Y),
\qquad
M_r(\eps)=1+\frac{2(1+Z)}L.
\]
Both comparisons are strictly multiplicative, including when one reference has zero cost.
\end{theorem}

Taking $r=2/\eps$ gives
\begin{equation}\label{eq:M0}
M_0(\eps)
=1+\frac{2\eps}{2+\eps}
 +\frac4{2+\eps}e^{2/\eps+1}
=(2e+o(1))e^{2/\eps}.
\end{equation}

\begin{corollary}[comparison with a competitive baseline]\label{cor:baseline}
Suppose a prediction-augmented online problem has a deterministic algorithm $X$ with $X=\OPT$ on
correct inputs and a deterministic strict $\rho$-competitive algorithm $Y$. If $(X,Y)$ satisfies
Definition~\ref{def:switching}, then for every $\eps\in(0,1]$ there is a randomised $(1+\eps)$-consistent,
$M_0(\eps)\rho$-robust algorithm.
\end{corollary}

\begin{proof}
Apply Theorem~\ref{thm:truncated-comparison} with $r=2/\eps$: comparison with $X$ proves consistency
on correct inputs, and comparison with $Y$ proves robustness on every input.
\end{proof}

\paragraph{Metric matching.}
Let $F_t^X,F_t^Y$ be the two residual labelled-server multisets and let $D_t$ be their minimum
matching distance. Pairing the servers consumed by $X$ and $Y$ on each request, and then taking
complements inside the common labelled server multiset, gives
\begin{equation}\label{eq:residual-distance}
D_t\le X_t+Y_t.
\end{equation}
An actual residual set can track either reference through a bijection: serving the server paired
with the reference server makes actual service plus the change in tracking potential at most the
reference service, while composing with a minimum matching between $F_t^X$ and $F_t^Y$ changes modes for
at most $D_t$. Thus every pair of deterministic online metric matching algorithms satisfies
Definition~\ref{def:switching} in every metric. Appendix~\ref{app:matching-switching} gives the full labelled-multiset proof.

Let $\mathrm{Pred}$ consume the predicted server whenever it is free, with an arbitrary fixed
fallback otherwise, so that it equals $\OPT$ on correct inputs; let $B_k$ be the deterministic
$\beta_k$-competitive algorithm on the real line from~\cite{Rag18}, where $\beta_k=O(\log k)$.

\begin{corollary}[uniform randomised upper bound]\label{cor:randfixedband}
For every $\eps\in(0,1]$ and every $k\ge2$,
\[
R_k^{\mathrm{rand},\R}(1+\eps)
\le
M_0(\eps)\beta_k
=e^{2/\eps}O(\log k).
\]
Together with Theorem~\ref{thm:robustness},
\[
R_k^{\mathrm{rand},\R}(1+\eps)
\le
\min\left\{
\frac{6^{k-1}}{\eps^{k-1}},
M_0(\eps)\beta_k
\right\}.
\]
\end{corollary}

\begin{proof}
Apply Corollary~\ref{cor:baseline} with $X=\mathrm{Pred}$ and $Y=B_k$.
\end{proof}

Let $\rho_k^0$ denote the infimum strict competitive ratio of a prediction-free randomised
algorithm on the real line against an oblivious adversary.

\begin{corollary}[comparison with the prediction-free ratio]\label{cor:symbolicrho}
For every $k\ge2$ and $\eps\in(0,1]$,
\[
 \rho_k^0\le R_k^{\mathrm{rand},\R}(1+\eps)
 \le M_0(\eps)\rho_k^0.
\]
Together with Theorem~\ref{thm:randuniformlb}, the available lower bound in the uniform range is
$\max\{\rho_k^0,\exp((2-o(1))/\eps)\}$. Whether the upper product is necessary remains open.
\end{corollary}

\begin{proof}
The lower bound feeds a fixed prediction sequence to any prediction-augmented algorithm. For the
upper bound, take a strict randomised baseline $B$ with expected ratio at most $\rho_k^0+\eta$,
sample its tape $\omega$ before the input, and condition on $\omega$ to obtain a deterministic
baseline. Apply Theorem~\ref{thm:truncated-comparison} to $(\mathrm{Pred},B_\omega)$ using fresh
randomness: comparison with $\mathrm{Pred}$ proves consistency after averaging over $\omega$,
while on every fixed input the other comparison and the tower property give
\[
\mathbb E[\cost]
\le M_0(\eps)\mathbb E[\cost(B)]
\le M_0(\eps)(\rho_k^0+\eta)\OPT.
\]
Letting $\eta$ decrease to zero proves the upper bound.
\end{proof}

\subsection{Applications beyond metric matching}\label{sec:beyondmatching}

The switching assumptions also hold outside irrevocable matching.

\begin{lemma}[metrical task systems]\label{lem:mtscoupling}
Every pair of deterministic algorithms for the same metrical task system, starting at the same
state, satisfies Definition~\ref{def:switching}.
\end{lemma}

\begin{proof}
Let the two reference states after prefix $t$ be $u_t,v_t$. Since both trajectories start at the
same state,
\[
d(u_t,v_t)\le\sum_{i\le t}d(u_{i-1},u_i)+\sum_{i\le t}d(v_{i-1},v_i)\le X_t+Y_t.
\]
If the actual state is $a$ and it tracks a reference at $u$, use $d(a,u)$ as the potential; when the
reference moves to $u'$ for task $\ell_t$, the actual algorithm also moves to $u'$, and
\[
d(a,u')+\ell_t(u')-d(a,u)\le d(u,u')+\ell_t(u').
\]
Changing the tracked reference from $u$ to $v$ increases the potential by at most $d(u,v)$.
\end{proof}

For MTS, Christianson, Shen and Wierman~\cite{CSW23} already give DART, a strict multiplicative,
diameter-free asymmetric combiner with robustness $2^{O(1/\eps)}$ relative to the competitive
reference; their published analysis gives a constant at most $4$ when the bound is written as
$\exp((C+o(1))/\eps)$. Lemma~\ref{lem:mtscoupling} verifies Definition~\ref{def:switching}
for MTS, and the lower and upper bounds under those assumptions give $C=2$. A classical
baseline satisfying only $\cost\le\alpha\OPT+b$ yields
$M_0(\eps)\alpha\OPT+M_0(\eps)b$, including the additive term, whereas the metric-matching
applications remain strict because they use strict finite-horizon baselines.

Thus, when $k$ is large and $\eps$ is fixed,
\[
\max\left\{
e^{(2-o(1))/\eps},\,\Omega(\sqrt{\log k})
\right\}
\le R_k^{\mathrm{rand},\R}(1+\eps)
\le M_0(\eps)O(\log k).
\]
The dependence on $\eps$ remains exponential, while the dependence on the number of servers is at
most logarithmic apart from the factor $M_0(\eps)$.

\section{Discussion and open problems}\label{sec:discussion}

\paragraph{Summary of bounds.}
For every fixed $k$, Corollary~\ref{cor:bracket} determines the deterministic
near-perfect-consistency trade-off as $\Theta_k(1/\eps^{k-1})$, and
Theorem~\ref{thm:randk} gives the same exponent for randomised algorithms. At fixed positive
$\eps$, Corollary~\ref{cor:randfixedband} caps randomised robustness at
$M_0(\eps)O(\log k)$, while Theorem~\ref{thm:randuniformlb} forces
$e^{(2-o(1))/\eps}$ once $k=\Theta(1/\eps)$, fixing the constant in the exponent at $2$ on this
scale apart from the separate prediction-free dependence on $k$. The least factor obtainable under
Definition~\ref{def:switching} has the same constant for every fixed $c>1$. Both deterministic
values are exact for $k=2$, and on the real line the complete expression $1+4/\eps+8/\eps^2$ for
$k=3$ is exact when $0<\eps\le\sqrt{13}-3$ and conjecturally throughout $0<\eps\le1$
(Theorem~\ref{thm:k3exact}, Conjecture~\ref{conj:k3}). Corollary~\ref{cor:separation} gives
exponential growth in $k$ for every fixed consistency factor
below $3$ and a lower bound $2k-1$ at factor $3$.

\paragraph{Open problems.}
\begin{itemize}\itemsep2pt
\item \emph{Constants.} At $k=2$ both deterministic values are known exactly, and on the real line the
leading $1/\eps^2$ coefficient at $k=3$ is $8$. Is the lower-bound value $L_k(\eps)$ exactly
optimal at finite $\eps$ or for larger $k$?
\item \emph{The three-server value above $\sqrt{13}-3$.} Theorem~\ref{thm:k3exact} determines
$R_3^{\R}(1+\eps)$ for $0<\eps\le\sqrt{13}-3$, and Conjecture~\ref{conj:k3} asserts the same value
throughout $0<\eps\le1$. By Remark~\ref{rem:k3open} what is missing is the residual branch of the
second-request analysis under the offline request--server assignments other than the one that
Proposition~\ref{prop:k3residual} settles.
\item \emph{Consistency factor three.} Is polynomial or subexponential robustness attainable at
consistency factor $3$ on the real line?
\item \emph{The randomised joint dependence.} The constant in the exponential dependence on
$1/\eps$ is $2$. Is the correct joint dependence a maximum, a sum, or a product
of this term and the classical randomised complexity on the real line (itself only bounded between
$[\Omega(\sqrt{\log k}),\,O(\log k)]$)?
\item \emph{Factor under the switching assumptions.} Corollary~\ref{cor:optimal-rate}
determines the exponential rate for every fixed $c>1$. What is the exact finite-$\eps$
loss, and what is the crossover when $c-1=\Theta(\eps)$?
\item \emph{General metrics.} Theorems~\ref{thm:consistency} and~\ref{thm:robustness} hold in
every metric. The lower-bound construction uses the order on the real line. In general metrics the deterministic
prediction-free floor is $2k-1$~\cite{KP93}: for which metrics does
$(1+\eps)$-consistency force robustness $\omega(2k-1)$?
\end{itemize}

\begin{remark}[Scope of the switching interface]
The product in Corollary~\ref{cor:symbolicrho} follows from the information available to the
comparison theorem. Definition~\ref{def:switching} exposes the baseline $Y$ through its cumulative
cost, so the theorem gives $M(\eps)\cost(Y)$. Substituting
$\cost(Y)\le\rho_k^0\OPT$ then multiplies the two factors. A maximum or sum bound requires an
interface that also uses the reason that $Y$ is competitive. The lower-bound family in
Appendix~\ref{app:lb} uses $\Theta(1/\eps)$ servers and treats additional servers as zero-cost
padding. It therefore does not combine its inverse-slack lower bound with prediction-free hardness in
$k$. Resolving the joint dependence requires a construction in which both constraints act on the same
requests.
\end{remark}

\ifanon\else
\paragraph{Acknowledgements.}
The author thanks Yichen Huang for advising this project and for comments on an earlier draft.
\fi

\appendix

\section{Full proof of the randomised lower bound}\label{app:lb}

This appendix proves the three claims used in the proof sketch of
Theorem~\ref{thm:randuniformlb}: a cost decomposition on the real line, a reduction of an arbitrary
randomised algorithm to a mixture indexed by the first use of $s_0$, and greedy minimisation subject to the resulting
prefix constraints. Together they prove
Theorem~\ref{thm:randuniformlb} against \emph{every} randomised $(1+\eps)$-consistent algorithm.

The constant $2$ in $e^{(2-o(1))/\eps}$ follows from the reduction in
Theorem~\ref{thm:canon} and the optimisation in Theorems~\ref{thm:greedy}
and~\ref{thm:sharpexp}; by comparison, the fixed-$k$ event bound in
Theorem~\ref{thm:randk} uses a union bound and yields a constant of order $1/e$ in the exponent.
The reduction is therefore required to obtain the constant $2$.

\paragraph{The input family.}
Fix $0<x<a_1<\dots<a_n$, place servers $s_0=0$ and $s_i=a_i$ for $1\le i\le n$, and issue the
common prefix $(r_1,p_1)=(x,s_1)$ and $(r_t,p_t)=(a_{t-1},s_t)$ for $t=2,\dots,n$ on these
$k=n+1$ servers. Write $b_1=x$ and $b_i=a_{i-1}$ for $i\ge2$. For $1\le j\le n$, the
\emph{perfect stopping instance} $P_j$ follows the first $j$ prefix requests with a request at
$0$ and the exact hits $a_{j+1},\dots,a_n$; Lemma~\ref{lem:uncross} gives
$\OPT(P_j)=a_j-x$, while the \emph{terminal instance} $B$ follows the whole prefix and ends with
a request at $a_n$, with optimum $x$.

For $1\le i\le n$, define the deterministic realisation $H_i$ to follow the predictions at stages
$1,\dots,i-1$, match $b_i$ to $s_0$ at stage $i$, and use $s_{t-1}$ at each later prefix stage
$t>i$. It matches the request at $0$ in $P_j$, or the final request at $a_n$ in $B$, to the
remaining server; let $H_\infty$ follow the predictions throughout the prefix. Telescoping gives
\[
\cost(H_i,P_j)=\begin{cases}a_j-x,& j<i,\\[2pt] a_j-x+2b_i,& j\ge i,\end{cases}
\qquad
\cost(H_i,B)=2b_i-x,
\]
and $\cost(H_\infty,P_j)=a_j-x$, $\cost(H_\infty,B)=2a_n-x$.

\begin{lemma}[cost decomposition on the real line]\label{lem:linecost}
For any perfect matching of requests to servers on the real line, orient each matched edge from its
request $r$ to its server $s$ and set $\Delta=\sum_{\mathrm{edges}}(s-r)$ and
$N=\sum_{\mathrm{edges}}\max\{r-s,0\}$. Then $\cost=\Delta+2N$. Consequently, on $P_j$ every
matching containing the edge $b_i\to0$ costs at least $a_j-x+2b_i$; on $B$ every matching
containing $b_i\to0$ costs at least $2b_i-x$, and every matching containing $a_n\to0$ costs at
least $2a_n-x$. Each bound is attained by the corresponding algorithm $H_i$ or $H_\infty$.
\end{lemma}

\begin{proof}
Let $P=\sum_{\mathrm{edges}}\max\{s-r,0\}$ and $N=\sum_{\mathrm{edges}}\max\{r-s,0\}$. Then
$\cost=P+N$, $\Delta=P-N$, and hence $\cost=\Delta+2N$, where the signed sum
$\Delta=\sum s-\sum r$ depends only on the two point multisets. On $P_j$ we have
$\Delta=a_j-x$, so a leftward edge $b_i\to0$ forces $N\ge b_i$ and gives
$\cost\ge a_j-x+2b_i$; on $B$ we have $\Delta=-x$, and the edges $b_i\to0$ and $a_n\to0$
force $N\ge b_i$ and $N\ge a_n$, respectively. Each corresponding algorithm has exactly the stated
leftward edge and no other leftward displacement, so equality holds.
\end{proof}

\begin{theorem}[reduction to first-use probabilities]\label{thm:canon}
Let $A$ be any randomised online matching algorithm on the family $\{P_1,\dots,P_n,B\}$, run under
a single common random tape. For $1\le i\le n$, let $q_i$ be the probability that $A$ first
consumes $s_0$ at prefix stage $i$, and let $q_\infty$ be the probability that it never consumes
$s_0$ during the prefix. Then, simultaneously for every $j$,
\[
\mathbb E[\cost(A,P_j)]\ge a_j-x+2\sum_{i\le j} q_i b_i,
\qquad
\mathbb E[\cost(A,B)]\ge\sum_{i} q_i(2b_i-x)+q_\infty(2a_n-x).
\]
The mixture that runs $H_i$ with probability $q_i$ and $H_\infty$ with probability
$q_\infty$ attains equality throughout. Hence every randomised algorithm is coordinatewise
dominated on this family by such a mixture.
\end{theorem}

\begin{proof}
Fix a realisation of the common-prefix tape, and suppose first that $A$ consumes $s_0$ at stage $i$.
On every stopping instance $P_j$ with $j\ge i$, the history through stage $i$ agrees with that of
$B$, so the completed matching contains $b_i\to0$ and Lemma~\ref{lem:linecost} gives cost at least
$a_j-x+2b_i$; for $j<i$, the bound $\OPT(P_j)=a_j-x$ suffices, while on $B$ the same edge gives
cost at least $2b_i-x$. If $A$ never consumes $s_0$ in the prefix, its $n$ prefix requests consume
the servers $s_1,\dots,s_n$, forcing the final request to $s_0$, and the resulting matching contains
$a_n\to0$ and costs at least $2a_n-x$, while each $P_j$ still costs at least $a_j-x$. Since $P_j$
and $B$ share their first $j$ requests, the event ``$s_0$ is first consumed at stage $i$'' has the
same probability on both inputs whenever $i\le j$, which makes the conditioning consistent across
the family. Averaging over the tape gives the two displayed inequalities, and the cost formulas for
$H_i$ and $H_\infty$ give equality for the mixture.
\end{proof}

If $A$ is $(1+\eps)$-consistent, the inequalities for $P_j$ in Theorem~\ref{thm:canon} become the
prefix constraints $2\sum_{i\le j}q_ib_i\le\eps(a_j-x)$, and the cost of $A$ on $B$ is at least
$K(q):=\sum_i q_i(2b_i-x)+q_\infty(2a_n-x)$. Define the incremental capacities
\[
c_1=\frac{\eps(a_1-x)}{2b_1}=\frac{\eps(1-x)}{2x},
\qquad
c_i=\frac{\eps(a_i-a_{i-1})}{2b_i}=\frac{\eps(a_i-a_{i-1})}{2a_{i-1}}\ \ (i\ge2).
\]

\begin{lemma}[consistency constraints for the mixture]\label{lem:polytope}
The mixture is $(1+\eps)$-consistent on every $P_j$ if and only if $q_i\ge0$,
$\sum_iq_i\le1$, and $2\sum_{i\le j}q_ib_i\le\eps(a_j-x)$ for all $j$.
\end{lemma}

\begin{proof}
By Theorem~\ref{thm:canon} the mixture attains cost exactly $\OPT(P_j)+2\sum_{i\le j}q_ib_i$ on
$P_j$, with $\OPT(P_j)=a_j-x$; the displayed inequalities are therefore precisely
$(1+\eps)$-consistency. Nonnegativity and total mass at most one are exactly the requirement that
$q_\infty=1-\sum_iq_i$ be a well-defined probability.
\end{proof}

\begin{theorem}[greedy minimisation]\label{thm:greedy}
Among all $q$ satisfying Lemma~\ref{lem:polytope}, the objective $K(q)$ is minimised by the
greedy vector $q^G$ obtained by scanning $i=1,2,\dots$ with remaining mass $z$ (initially $1$) and
setting $q_i^G=\min\{z,c_i\}$, $z\leftarrow z-q_i^G$.
\end{theorem}

\begin{proof}
\emph{Feasibility.} For every prefix $j$, before the mass is exhausted,
$2\sum_{i\le j}q_i^Gb_i\le2\sum_{i\le j}c_ib_i=\eps\bigl[(a_1-x)+\sum_{i=2}^j(a_i-a_{i-1})\bigr]
=\eps(a_j-x)$; after the mass is exhausted the later coordinates vanish and the same prefix
inequality persists.

\emph{Optimality.} Write $q_\infty=1-\sum_iq_i$, so that
$K(q)=2a_n-x-\sum_i q_i\cdot2(a_n-b_i)$. Minimising $K$ is therefore equivalent to maximising the
linear objective $\sum_i q_i\,2(a_n-b_i)$, whose coefficients $2(a_n-b_i)$ are nonincreasing in $i$
because $b_1\le b_2\le\dots\le b_n<a_n$. Let $q$ be feasible, and suppose some positive mass sits at
a coordinate $j$ while an earlier prefix budget is not tight, so that for some $i<j$ and some
$\delta>0$ moving $\delta$ mass from $j$ to $i$ preserves every prefix constraint. This move changes
the objective by $2\delta\bigl[(a_n-b_i)-(a_n-b_j)\bigr]=2\delta(b_j-b_i)\ge0$, strictly when
$b_j>b_i$. Formally, take the first coordinate $i$ at which an optimum differs from the greedy
vector; all earlier greedy-saturated prefix constraints then agree. If the optimum has less mass
at $i$, either mass remains at infinity or some later coordinate $j$ has excess mass, so move the
largest admissible $\delta$ from that later coordinate (or infinity) to $i$. The earlier constraints
are unchanged, feasibility between $i$ and $j-1$ follows from the first slack prefix, and from
$j$ onward the weighted load weakly decreases because $b_i\le b_j$, so the objective weakly
improves. Repetition makes coordinate $i$ greedy, and induction fixes all coordinates. When all
earlier prefix budgets are tight, subtracting the $(j-1)$-st budget from the $j$-th shows that the
additional admissible mass is exactly $c_j$ (and $c_1$ at $j=1$). Hence the optimum is the greedy
vector; with strictly increasing scales from stage $2$ on, its value of $K$ is unique.
\end{proof}

\begin{theorem}[geometric lower bound]\label{thm:sharpexp}
Fix $x=\tfrac12$ and $\delta>0$, take $a_i=(1+\delta)^{i-1}$, and set
$L=\bigl\lfloor 2(1-\eps/2)/(\eps\delta)\bigr\rfloor$ and $n=L+2$. Then every
$(1+\eps)$-consistent mixture of the algorithms $H_i,H_\infty$ has cost on $B$ satisfying
\[
K(q)\ge K(q^G)\ge\eps\bigl[(1+\delta)^L-1\bigr]-\tfrac12.
\]
\end{theorem}

\begin{proof}
With $x=\tfrac12$ and $a_1=1$ the capacities are $c_1=\eps(1-x)/(2x)=\eps/2$ and
$c_i=\eps(a_i-a_{i-1})/(2a_{i-1})=\eps\delta/2$ for $i\ge2$. The greedy rule spends $c_1=\eps/2$ at
coordinate $1$ and then fills coordinates $2,\dots,L+1$ completely, since
$c_1+L(\eps\delta/2)\le1$ by the choice of $L$. By Theorem~\ref{thm:greedy}, dropping the
nonnegative coordinate-$1$ and any final-partial contributions,
\[
K(q^G)\ \ge\ \sum_{i=2}^{L+1}\frac{\eps\delta}{2}\,(2a_{i-1}-x)
\ =\ \eps\delta\sum_{t=0}^{L-1}(1+\delta)^t-\frac{\eps\delta x}{2}L
\ \ge\ \eps\bigl[(1+\delta)^L-1\bigr]-x,
\]
where $\eps\delta\sum_{t=0}^{L-1}(1+\delta)^t=\eps\bigl[(1+\delta)^L-1\bigr]$ and
$(\eps\delta/2)L\le1$ gives $(\eps\delta x/2)L\le x=\tfrac12$.
\end{proof}

\begin{proof}[Proof of Theorem~\ref{thm:randuniformlb}]
Given $\eta>0$, choose $\delta=\delta(\eta)>0$ small enough that
$2\log(1+\delta)/\delta\ge2-\eta/3$. Instantiate the construction of Theorem~\ref{thm:sharpexp}. By
Theorem~\ref{thm:canon}, every randomised $(1+\eps)$-consistent algorithm on this family is
coordinatewise dominated by a $(1+\eps)$-consistent mixture of the algorithms $H_i,H_\infty$; by
Theorems~\ref{thm:greedy}--\ref{thm:sharpexp} the mixture's cost on $B$ is at least
$\eps[(1+\delta)^L-1]-\tfrac12$. Dividing by $\OPT(B)=x=\tfrac12$, the robustness ratio is at least
$2\eps(1+\delta)^L-2$. Now
\[
L\log(1+\delta)\ \ge\ \Bigl[\tfrac{2(1-\eps/2)}{\eps\delta}-1\Bigr]\log(1+\delta)
\ \ge\ \frac{2-\eta/3}{\eps}-O_\eta(1),
\]
so $2\eps(1+\delta)^L\ge\exp\!\bigl(\tfrac{2-\eta/3}{\eps}+\log(2\eps)-O_\eta(1)\bigr)$. Since
$1/\eps$ dominates $|\log\eps|$, for all sufficiently small $\eps$ the additive $-2$ and the
prefactor loss $\log(2\eps)-O_\eta(1)$ together cost at most $2\eta/(3\eps)$ in the exponent, so the
ratio is at least $\exp((2-\eta)/\eps)$. The instance uses $k_0\le C_\eta/\eps$ servers.
For $k>k_0$, add $k-k_0$ labelled servers at new pairwise distinct locations, also distinct from
all locations in the construction, and prepend their truthfully predicted exact-hit requests in a
fixed order. Compare with the fixed completion that requests every still-unrequested server at
its own location. This input has
$\OPT=0$, so finite expected robustness forces zero total cost almost surely; location uniqueness
therefore forces the intended padding label at every padding step. Prefix indistinguishability
transfers these actions to every constructed suffix. The prepended zero edges preserve $\OPT$ and extend
every perfect stopping instance to a perfect $k$-server instance. All inputs are fixed independently
of the random tape, so the padding is oblivious.
\end{proof}

\section{Full proof of the comparison theorem}\label{app:comparison}

This appendix proves Theorem~\ref{thm:comparison} using only Definition~\ref{def:switching},
without using a matching-specific property. Let $X,Y$ be a deterministic pair satisfying that definition,
with per-request costs $x_t,y_t$ and cumulative costs
$X_t,Y_t$.  Recall
\[
F(z)=\min\{1,\log(1+z)/L\},\qquad Z=e^L-1,
\]
and
\[
a(z)=1-F(z)+(1+z)F'(z),\qquad
b(z)=F(z)+z(1+z)F'(z).
\]
On the axes the formulas use their continuous extensions. At a zero-cost prefix the algorithm
retains its current mode while both new reference costs are zero; on the first request with
positive reference mass, it observes $(x,y)$, switches for free because $D_0=0$ to the Bernoulli
marginal $F(x/y)$ with $F(\infty)=1$, and serves once in that mode. Thereafter it uses the ordinary
rule with a switch before and after service, while for $x=y=0$ reference tracking and
nonnegativity of $\Psi$ force zero actual cost.

\begin{lemma}[interpolation]\label{lem:interp}
The switching assumptions implement the interpolated mode process as a feasible randomised
online algorithm. Its expected actual cost is at most the sum over requests of
\eqref{eq:interpolation}, and this line integral is at most
$\int[a(z)\,dX+b(z)\,dY]$.
\end{lemma}

\begin{proof}
For request $t$, write $A=X_{t-1}$, $B=Y_{t-1}$, $x=x_t$, $y=y_t$, and let
$A(s)=A+sx$, $B(s)=B+sy$, $q(s)=F(A(s)/B(s))$ for $s\in[0,1]$, with
$q_0=q(0)$, $q_1=q(1)$ and $\bar q=\int_0^1q(s)\,ds$. The ratio $A(s)/B(s)$ is monotone because
its derivative has the constant sign of $xB-yA$.

The current mode has marginal $q_0$ of tracking $Y$. Before service, maximally couple it to a
Bernoulli variable of marginal $\bar q$, and after service couple again to marginal $q_1$, giving
switch probabilities $|\bar q-q_0|$ and $|q_1-\bar q|$. Property (S) of
Definition~\ref{def:switching} bounds expected service cost plus the change in tracking
potential by $(1-\bar q)x+\bar qy$, while property (T) bounds the two expected switch increases by
$(A+B)|\bar q-q_0|$ and $(A+B+x+y)|q_1-\bar q|$.

Since $q$ is monotone, Fubini gives
\[
|\bar q-q_0|=\int_0^1(1-s)|dq(s)|,
\qquad
|q_1-\bar q|=\int_0^1s|dq(s)|.
\]
Consequently the expected amortised cost on this request is at most
\[
\int_0^1\!\bigl[(1-q(s))x+q(s)y+(A(s)+B(s))|dq(s)|\bigr],
\]
which is exactly the contribution of \eqref{eq:interpolation}.  This is an implementable discrete
algorithm: the interpolation is used only to choose the two Bernoulli marginals; the request itself
is served once, in one feasible mode.

For the second inequality, at differentiability points, with $z=X/Y$,
\[
(X+Y)|dz|=(1+z)|dX-z\,dY|
 \le (1+z)dX+z(1+z)dY,
\]
because $dX,dY\ge0$. Since $|dF|=F'(z)|dz|$, adding the service terms gives
$a(z)dX+b(z)dY$; continuity of $F$ means that the cap $z=Z$ carries no atom. The tracking
potential starts at zero and is nonnegative at termination, so dropping its terminal value
converts the amortised bound into a bound on actual expected cost.
\end{proof}

Fix $c\ge1$ and recall
$\kappa(c)=2\log c+(c+1)\log(1+1/c)$.

\begin{lemma}[potential for comparison with $X$]\label{lem:nearX}
Let
\[
\phi(z)=\frac{2z\log z-(z+1)\log(1+z)+\kappa(c)z}{L}
\quad(0\le z\le Z),
\]
with $0\log0=0$, and extend it for $z\ge Z$ as
$\phi(z)=\alpha z-1$, where $\alpha=(\phi(Z)+1)/Z$. Put $V(X,Y)=Y\phi(X/Y)$.
Then
\[
d(\mathrm{meta\text{-}cost})+dV
 \le \left(1+\frac{2+\kappa(c)}L\right)dX,
\]
and $V\ge0$ at termination whenever $Y\le cX$.
\end{lemma}

\begin{proof}
For $0<z<Z$,
\[
F(z)=\frac{\log(1+z)}L,\qquad F'(z)=\frac1{L(1+z)},
\]
so
\[
a(z)=1-\frac{\log(1+z)}L+\frac1L,
\qquad
b(z)=\frac{\log(1+z)+z}{L}.
\]
Differentiating $\phi$ gives
\[
\phi'(z)=\frac{2\log z+1-\log(1+z)+\kappa(c)}L,
\qquad
\phi(z)-z\phi'(z)=-b(z),
\]
and
\[
a(z)+\phi'(z)
=1+\frac{2+\kappa(c)+2\log(z/(1+z))}{L}
\le1+\frac{2+\kappa(c)}L.
\]
For $V=Y\phi(X/Y)$,
$dV=\phi'(z)dX+[\phi(z)-z\phi'(z)]dY$.  Lemma~\ref{lem:interp} therefore cancels the $dY$
coefficient below the cap.

Above the cap, $a=0$, $b=1$, and the linear extension has
$\phi-z\phi'=-1=-b$.  Its slope satisfies
\[
\alpha=\frac{h(Z)+\kappa(c)}L+\frac1Z
\le1+\frac{\kappa(c)}L
\le1+\frac{2+\kappa(c)}L,
\]
where $h(z)=2\log z-(1+1/z)\log(1+z)$ and we used
$h(Z)+L/Z=2\log Z-L\le L$.

For the terminal sign, $h'(z)=[z+\log(1+z)]/z^2>0$ and $\kappa(c)=-h(1/c)$, so below the cap
$\phi(z)/z=[h(z)-h(1/c)]/L\ge0$ whenever $z\ge1/c$; the linear extension remains nonnegative as
well. Hence $Y\le cX$ implies $V\ge0$ at termination, completing the differential inequality in
both regions.
\end{proof}

\begin{lemma}[potential for comparison with $Y$]\label{lem:nearY}
Let $\chi(z)=\int_z^Za(s)\,ds$ for $0\le z\le Z$, and $\chi(z)=0$ for $z\ge Z$.
Put $W(X,Y)=Y\chi(X/Y)$. Then
\[
d(\mathrm{meta\text{-}cost})+dW\le M\,dY,
\qquad M=1+\frac{2(e^L-1)}L,
\]
and $W$ is nonnegative and initially zero.
\end{lemma}

\begin{proof}
Below the cap, $\chi'=-a$, so the $dX$ coefficient cancels. The remaining $dY$ coefficient is
\[
g(z)=b(z)+\chi(z)+z a(z).
\]
Using $b'(z)=(2+z)/[L(1+z)]$ and $a'(z)=-1/[L(1+z)]$ gives
$g'(z)=2/[L(1+z)]>0$, so $g$ is maximised at $Z$, where
\[
g(Z)=b(Z)+Za(Z)=1+\frac{2Z}{L}=M.
\]
Above the cap the coefficient is $1\le M$, while $a\ge0$ below the cap implies $\chi\ge0$ and
$W\ge0$.
\end{proof}

\begin{proof}[Proof of Theorem~\ref{thm:comparison}]
Lemma~\ref{lem:nearY} gives $\mathbb E[\cost]\le M\cost(Y)$ on every input. If
$\cost(Y)\le c\cost(X)$, Lemma~\ref{lem:nearX} and
$L=(2+\kappa(c))/\eps$ give
\[
\mathbb E[\cost]\le
\left(1+\frac{2+\kappa(c)}L\right)\cost(X)
=(1+\eps)\cost(X).
\]
Lemma~\ref{lem:interp} supplies the feasible randomised online implementation, with the input fixed
independently of its random bits as required by the oblivious adversary model.
\end{proof}

\section{Analysis of the truncated logarithm}\label{app:truncated}

\begin{proof}[Proof of Theorem~\ref{thm:truncated-comparison}]
Let
\[
F_r(z)=
\begin{cases}
0,&0\le z\le r,\\[1mm]
\log(z/r)/L,&r<z<Z,\\[1mm]
1,&z\ge Z.
\end{cases}
\]
Use the maximal Bernoulli couplings from Lemma~\ref{lem:interp}. Write $f=F_r'$ where the
derivative exists, and put
\[
a(z)=1-F_r(z)+(1+z)f(z),
\qquad
b(z)=F_r(z)+z(1+z)f(z).
\]
Lemma~\ref{lem:interp} gives
\begin{equation}\label{eq:truncatedab}
d\mathbb E[\cost]\le a(z)\,dX+b(z)\,dY,
\qquad z=\frac XY.
\end{equation}
The coefficient functions are
\[
(a(z),b(z))=
\begin{cases}
(1,0),&0\le z<r,\\[1mm]
\left(1-\dfrac{\log(z/r)}L+\dfrac{1+1/z}L,
\dfrac{\log(z/r)+1+z}L\right),&r<z<Z,\\[4mm]
(0,1),&z>Z.
\end{cases}
\]
Continuity of $F_r$ means that crossing an endpoint creates no jump charge.

For the comparison with $X$, define
\[
\phi(z)=z\int_0^z\frac{b(s)}{s^2}\,ds
\qquad(0<z\le Z).
\]
For $V(X,Y)=Y\phi(X/Y)$, the identity $\phi-z\phi'=-b$ gives
\begin{equation}\label{eq:truncatedXpotential}
d\mathbb E[\cost]+dV
\le\bigl(a(z)+\phi'(z)\bigr)dX.
\end{equation}
Below $r$, this coefficient is $1$, while on $(r,Z)$,
\[
\frac{d}{dz}\bigl(a(z)+\phi'(z)\bigr)
=a'(z)+\frac{b'(z)}z
=2(1+z)\left(f'(z)+\frac{f(z)}z\right)=0,
\]
because $f(z)=1/(Lz)$, and its right limit at $r$ is
\[
1+\frac{2(1+1/r)}L=1+\eps.
\]
Above $Z$, extend $\phi$ linearly as $\phi(z)=mz-1$, where
$m=(\phi(Z)+1)/Z$. Direct integration gives
\[
m=1+\frac2L\left(\frac1r-\frac1Z\right)<1+\eps.
\]
Since the potential is nonnegative and initially zero, integrating
\eqref{eq:truncatedXpotential} proves
$\mathbb E[\cost]\le(1+\eps)X$.

For the comparison with $Y$, define
\[
\chi(z)=\int_z^Z a(s)\,ds
\quad(0\le z\le Z),
\qquad
\chi(z)=0
\quad(z\ge Z),
\]
and set $W(X,Y)=Y\chi(X/Y)$. Since $a\ge0$, the potential is nonnegative; below the cap,
$\chi'=-a$, so the coefficient of $dX$ in $d\mathbb E[\cost]+dW$ vanishes, and the remaining
coefficient is
\[
g(z)=b(z)+\chi(z)+za(z).
\]
On each differentiability interval,
\[
g'(z)=b'(z)+za'(z)
=2(1+z)\frac{d}{dz}\bigl(zf(z)\bigr).
\]
Thus $g$ is constant below $r$ and on $(r,Z)$, and on the middle interval its value is
\[
g(Z^-)=b(Z^-)+Za(Z^-)
=1+\frac{2(1+Z)}L
=M_r(\eps).
\]
The coefficient is smaller below $r$ and equals $1$ above $Z$, with the same bounds on the axes.
Since the potential starts at zero, integration proves
$\mathbb E[\cost]\le M_r(\eps)Y$. The interpolation is online by Lemma~\ref{lem:interp}.
\end{proof}

\section{Analysis of the shifted logarithm}\label{app:shifted}

\begin{proof}[Proof of Theorem~\ref{thm:shifted-comparison}]
Use Lemma~\ref{lem:interp} with
$F(z)=\min\{1,\log(1+\eps z)/L\}$ and cap $Z=(e^L-1)/\eps$; the explicit zero-prefix rule above
handles the origin.  Below the cap put
\[
A(z)=1-F(z)+(1+z)F'(z),\qquad
B(z)=F(z)+z(1+z)F'(z).
\]
Lemma~\ref{lem:interp} bounds expected amortised cost by $A(z)dX+B(z)dY$.

For the comparison to $X$, let $z_0=1/c$ and define
\[
 \phi(z)=z\int_{z_0}^{z}\frac{B(s)}{s^2}\,ds\qquad(0<z\le Z).
\]
Thus $\phi-z\phi'=-B$, so $V(X,Y)=Y\phi(X/Y)$ cancels the $dY$ coefficient.  Direct
differentiation gives
\[
A(z)+\phi'(z)=1+\frac{G(z)}L,
\]
where, with
\[
 I(z)=-\frac{\log(1+\eps z)}z+2\eps\log z+(1-2\eps)\log(1+\eps z),
\]
\[
G(z)=2\eps\log z-2\eps\log(1+\eps z)
 +\frac{2\eps(1+z)}{1+\eps z}-I(z_0).
\]
Also,
\[
G'(z)=\frac{2\eps(1+z)}{z(1+\eps z)^2}>0,
\qquad \lim_{z\to\infty}G(z)=K.
\]
Consequently $A+\phi'\le1+K/L=1+\eps$ below the cap. Above it, extend
$\phi(z)=mz-1$, where $m=(\phi(Z)+1)/Z$; substitution of
$\eps Z=e^L-1$ shows $m\le1+K/L$ (equivalently,
$\eps\log(1-e^{-L})\le1$), while $\phi-z\phi'=-1=-B$. Also,
\[
K=2+2\eps\log(1+c/\eps)+(c-1)\log(1+\eps/c)>2,
\]
so $Z\ge z_0$. Since $B\ge0$, the integral definition gives $\phi(z)\ge0$ for
$z\in[z_0,Z]$, and the continuous linear extension remains nonnegative; hence whenever
$Y\le cX$, integration from the zero initial potential gives
$\mathbb E[\cost]\le(1+\eps)X$.

For the comparison to $Y$, define $\chi(z)=\int_z^Z A(s)ds$ below the cap and $0$ above it,
and set $W(X,Y)=Y\chi(X/Y)$. The $dX$ coefficient cancels, and the remaining coefficient
$g(z)=B(z)+\chi(z)+zA(z)$ satisfies
\[
g'(z)=\frac{2\eps(1+z)}{L(1+\eps z)^2}>0,
\qquad
g(Z)=1+\frac{2\eps Z(1+Z)}{Le^L}=\widehat M(c,\eps).
\]
Above the cap it is $1$, so the fact that $W$ starts at zero and ends nonnegative proves the
unconditional comparison with $Y$.

For fixed $c$, $K=2+O_c(\eps\log(1/\eps))$, hence
$L=2/\eps+O_c(\log(1/\eps))$.  The displayed formula for $\widehat M$ then gives
$\log\widehat M=L-\log\eps-\log L+O(1)=2/\eps+O_c(\log(1/\eps))$.
\end{proof}

Define $M^\star(c,\eps)$ as the infimum of the factors $M$ for which a randomised online algorithm
can satisfy $\mathbb E[\cost]\le M\cost(Y)$ for every deterministic pair obeying
Definition~\ref{def:switching}, together with the $(1+\eps)$ comparison to $X$ on every
input with $Y\le cX$.

\begin{corollary}[asymptotic value of $M^\star$]\label{cor:optimal-rate}
For every fixed $c>1$,
\[
\lim_{\eps\downarrow0}\eps\log M^\star(c,\eps)=2.
\]
\end{corollary}

\begin{proof}
Theorem~\ref{thm:shifted-comparison} gives the upper limit. We prove the lower limit directly
from the arbitrary-scale family of Appendix~\ref{app:lb}.

Fix $c>1$ and $\zeta\in(0,1)$.  Choose $\delta\in(0,1)$ so small that
$2\log(1+\delta)/\delta\ge2-\zeta/6$, and put
\[
 \lambda=\max\{1+\delta,\,2/(c-1)+2\}.
\]
For all sufficiently small $\eps>0$, so in particular $\eps\le1/2$ and
$2-\eps\lambda>0$, define
\[
 L'=\left\lfloor\frac{2-\eps\lambda}{\eps\delta}\right\rfloor,
 \qquad n=L'+3,
\]
and instantiate the family of Appendix~\ref{app:lb} with
\[
 x=\tfrac12,\qquad a_1=1,\qquad
 a_i=\lambda(1+\delta)^{i-2}\quad(2\le i\le n).
\]
Let $X=\mathrm{Pred}$. Define a deterministic algorithm $Y$ on all inputs as follows. On a prefix
of the constructed family, it follows the prediction at stage one, selects $s_0$ at stage two, and
selects $s_{t-1}$ at every later stage $t$. After the input first departs from this pattern, it
selects a nearest free server on every remaining request. Every pair of
deterministic matching algorithms satisfies Definition~\ref{def:switching} by
Lemmas~\ref{lem:matching-distance}--\ref{lem:tracking}.

For the perfect stopping instances $P_j$, the formulas in Appendix~\ref{app:lb} give
\[
 \cost(X,P_j)=\OPT(P_j)=a_j-x,
 \qquad
 \cost(Y,P_1)=a_1-x,
 \qquad
 \cost(Y,P_j)=a_j-x+2\quad(j\ge2).
\]
For $j\ge2$, $a_j\ge\lambda$, and therefore
\[
 (c-1)(a_j-x)\ge(c-1)\left(\frac2{c-1}+\frac32\right)>2.
\]
Thus $\cost(Y,P_j)\le c\cost(X,P_j)$ for every $j$. On the terminal instance $B$,
\[
 \cost(Y,B)=2a_1-x=\tfrac32,\qquad \OPT(B)=x=\tfrac12.
\]

Consider any algorithm included in the definition of $M^\star(c,\eps)$ with factor $M$. Since
$\cost(Y,P_j)\le c\cost(X,P_j)$, its comparison with $X$ makes the induced matching algorithm
$(1+\eps)$-consistent on all the instances $P_j$. By Theorems~\ref{thm:canon}
and~\ref{thm:greedy}, its expected cost on $B$ is at least the value of this mixture.
The incremental capacities are
\[
 c_1=\frac\eps2,\qquad c_2=\frac{\eps(\lambda-1)}2,\qquad
 c_i=\frac{\eps\delta}2\quad(i\ge3).
\]
The definition of $L'$ gives $c_1+c_2+L'\eps\delta/2\le1$, so greedy fills coordinates
$3,\ldots,L'+2$ completely.  Dropping all other nonnegative contributions yields
\begin{align*}
 \mathbb E[\cost(B)]
 &\ge \sum_{i=3}^{L'+2}\frac{\eps\delta}{2}(2a_{i-1}-x)\\
 &=\eps\lambda\bigl((1+\delta)^{L'}-1\bigr)
   -\frac{\eps\delta x}{2}L'\\
 &\ge\eps\bigl((1+\delta)^{L'}-1\bigr)-\frac12.
\end{align*}
After division by $\OPT(B)=1/2$, the ratio on $B$ is at least
$2\eps(1+\delta)^{L'}-2$. Also,
\[
 L'\log(1+\delta)
 \ge \frac{2-\zeta/6}{\eps}
     -\lambda\frac{\log(1+\delta)}\delta-\log2.
\]
Since $1/\eps$ dominates $|\log\eps|$ and the remaining terms depend only on $c,\zeta$, for all
sufficiently small $\eps$ this ratio is at least
$3\exp((2-\zeta)/\eps)$. The unconditional comparison with $Y$ bounds the same
ratio by $M\cost(Y,B)/\OPT(B)=3M$.  Hence
$M\ge\exp((2-\zeta)/\eps)$.  Letting $\zeta\downarrow0$ proves the lower limit.
\end{proof}

At the endpoint, if $c\le1+\eps$, always following $Y$ gives $M^\star(c,\eps)=1$; the lower
bound $M\ge1$ follows, for example, from a one-request problem with unavoidable positive cost and
$X=Y$. The transition between the two cases occurs when $c-1=\Theta(\eps)$.

\section{Verification of the switching assumptions for metric matching}\label{app:matching-switching}

We verify Definition~\ref{def:switching} for the matching application used in
Corollaries~\ref{cor:randfixedband} and~\ref{cor:symbolicrho}.  Server copies are labelled, so all
sets below are labelled multisets.

\begin{lemma}[prefix matching distance]\label{lem:matching-distance}
Let $X,Y$ be deterministic online matching algorithms on the same instance, and let
$F_t^X,F_t^Y$ be their residual free-server multisets after prefix $t$.  Their minimum matching
distance satisfies $D(F_t^X,F_t^Y)\le X_t+Y_t$.
\end{lemma}

\begin{proof}
Pair the server consumed by $X$ on request $i$ with the server consumed by $Y$ on the same request.
The triangle inequality bounds this pair by $x_i+y_i$, so the consumed multisets have minimum
matching cost at most $X_t+Y_t$.

It remains to pass to complements inside the common initial labelled multiset $S$.  For equal-size
labelled submultisets $A,B\subseteq S$, minimum matching distance satisfies
$D(A,B)=D(S\setminus A,S\setminus B)$.  Indeed, common labels may be matched to themselves in an
optimum: if a common label is sent elsewhere and another edge enters it, shortcut the two edges by
the triangle inequality, and iterate.  The optimum then matches only the two symmetric differences;
complementation swaps those differences, without changing their minimum matching distance. Applying this
to the consumed multisets gives the claim for the free multisets.
\end{proof}

\begin{lemma}[reference tracking for matching]\label{lem:tracking}
A feasible matching algorithm can track either reference with a nonnegative tracking
potential. Serving the server paired with the reference server satisfies property (S), and composing the
bijection with a minimum matching between the two reference residual sets satisfies
property (T).
\end{lemma}

\begin{proof}
Maintain a bijection $\pi$ from the actual free multiset to the current reference free multiset,
with potential $\Psi=\sum_a d(a,\pi(a))$.  If the reference serves request $r$ using $s$, the
combined algorithm serves it using $a=\pi^{-1}(s)$ and deletes the pair, after which
\[
d(r,a)\le d(r,s)+d(s,a),
\]
while the potential decreases by $d(s,a)$, proving (S). To change references, compose $\pi$ with a
minimum matching between their residual sets. This composition changes no previous assignment,
and the triangle inequality increases $\Psi$ by at most the matching distance, which is at most
$X_t+Y_t$ by Lemma~\ref{lem:matching-distance}; hence it proves (T).
\end{proof}

{\small
\bibliographystyle{plainurl}
\bibliography{refs}
}

\end{document}